\documentclass[a4paper,12pt]{article}
    \usepackage[top=2cm,bottom=2cm,left=2.5cm,right=2.5cm]{geometry}
    \usepackage{cite, amsmath, amssymb}
    \usepackage[margin=0.5cm,%
                font=small,%
                format=hang,%
                labelsep=period,%
                labelfont=bf]{caption}
\usepackage{amsmath}
\usepackage{amsfonts}
\usepackage{amssymb}
\usepackage{cite}
\usepackage{amsthm}
\usepackage{makeidx}
\usepackage{graphicx}
\usepackage{mathrsfs,xcolor,multicol}
\usepackage{enumerate}
\usepackage{blkarray}
\usepackage{bm}
\usepackage{setspace}
\usepackage{float}
\usepackage{authblk}
\usepackage{multirow}
\usepackage{booktabs}
\usepackage[ruled,vlined]{algorithm2e}
\usepackage{blkarray}
\usepackage{authblk}
\usepackage{multirow}
\usepackage{lineno}
\usepackage{array}
\usepackage{hyperref}
\usepackage{lscape}
\usepackage{spverbatim}

\usepackage{blkarray, bigstrut}
\usepackage{authblk}
\usepackage{multirow}

\usepackage{fancyhdr}
\usepackage{lastpage}

\usepackage{amsmath, amssymb, booktabs, lastpage,}
\usepackage{multicol}
\usepackage{hyperref}
\usepackage{blkarray}
\usepackage{empheq} 

\numberwithin{table}{section}
\numberwithin{equation}{section}

\theoremstyle{plain}
\newtheorem{theorem}{Theorem}[section]
\newtheorem{proposition}[theorem]{Proposition}
\newtheorem{definition}[theorem]{Definition}
\newtheorem{lemma}[theorem]{Lemma}
\newtheorem{example}[theorem]{Example}

\newtheorem{corollary}[theorem]{Corollary}
\newtheorem{remark}[theorem]{Remark}

\usepackage{tikz}
\usetikzlibrary{calc, arrows}
\newcommand{\matindex}[1]{\mbox{\scriptsize#1}}

\usepackage{tikz-cd}
\usetikzlibrary{circuits.logic.US,circuits.logic.IEC,fit}

\usetikzlibrary{arrows}

\newcommand{\NN}{\mathcal{N}} 

\usepackage{authblk}
\author[1,*]{ \textbf{Bryan S. Hernandez}}
\author[1]{ \textbf{Juan Paolo C. Santos}}
\author[2]{ \textbf{Patrick Vincent N. Lubenia}}
\author[2,3,4]{\textbf{Eduardo R. Mendoza}}

\affil[1]{\small \textit{Institute of Mathematics, University of the Philippines Diliman, Quezon City 1101, Philippines}}
\affil[2]{\small \textit{Center for Natural Science and Environmental Research, De la Salle University, Taft Avenue, Manila 0922, Philippines}}
\affil[3]{\small \textit{Department of Mathematics and Statistics, De La Salle University, Taft Avenue, Manila 0922, Philippines}}
\affil[4]{\small \textit{Max Planck Institute of Biochemistry, Martinsried near Munich 82152, Germany}}
\affil[*]{Email address: \texttt{bshernandez@up.edu.ph}}

\title{\textbf{The finest decompositions' architecture of a reaction network
}}
\date{}

\begin{document}
\maketitle
\begin{abstract} 
Biochemical and environmental modeling typically relies on reaction networks to represent complex transformations. While the Linkage Class Decomposition (LCD) partitions networks based on visual standard connectivity, it often misaligns with the algebraic properties governing long-term dynamics. This work establishes the Finest Decompositions’ Architecture (FDA) framework by analyzing hierarchical relationships between the LCD and two algebraic structures: the Finest Independent Decomposition (FID) and the Finest Incidence-Independent Decomposition (FIID).
These algebraic decompositions serve as the respective building blocks for characterizing general equilibria and complex-balanced equilibria of a reaction network.
Under the partial order of ``coarsens to,'' we categorize reaction networks into six architectures, distinguishing three subclasses of Independent Linkage Classes (ILC) from three subclasses of Dependent Linkage Classes (DLC). To facilitate the classification, we introduce the Deficiency Difference ($\Delta$), measuring the discrepancy between total and subnetwork deficiencies, and the Common Complexes Cardinality $|\mathscr{CC}|$ of the FID. Results show that $\Delta$ uniquely identifies all the ILC classes and one DLC subclass, while $|\mathscr{CC}|$ distinguishes the remaining DLC subclasses. We then provide numerous examples and identify interesting subclasses in the FDA classes. In particular, we prove that the class of deficiency zero networks is a proper subset of the FDE (Finest Decompositions Equality) class, where $\text{FID} = \text{FIID}$. A number of results on mass action systems such as the Deficiency One Theorem as well as on power law systems essentially rely on the ILC property of the underlying networks. These suggest that the FDA classification of ILC and DLC networks signify a certain alignment of both structural and kinetic attributes. This work opens up direction for the study of the structure and equilibria analysis of reaction networks across diverse decomposition architectures.

 
\end{abstract}

\thispagestyle{empty}

\section{Introduction}
\label{intro}

The study of complex systems, ranging from the regulation of blood sugar in the human body to the global transport of carbon in the environment of the Earth, necessitates robust mathematical modeling. These systems are frequently represented as Chemical Reaction Networks (CRNs)—essential maps detailing how substances (or species) interact and transform through chemical reactions. 
The standard reaction network diagram is partitioned into connected components associated with its linkage classes. In other words, the visual representation of a reaction network is fundamentally defined by its \textbf{Linkage Class Decomposition (LCD)}.
However, there is a frequent misalignment between a network's visual connectivity (LCD) and the underlying algebraic decompositions that govern its long-term dynamical behavior.

This paper addresses this gap by studying the relationships between a reaction network’s three finest decompositions:
\begin{itemize}
    \item[1.] \textbf{Linkage Class Decomposition (LCD)},
    \item[2.] \textbf{Finest Independent Decomposition (FID):} The unique independent decomposition with no independent refinement, which is critical for finding general equilibria,
    \item[3.] \textbf{Finest Incidence-Independent Decomposition (FIID):} The unique incidence-independent decomposition with no incidence-independent refinement, which is essential for finding complex-balanced equilibria.
\end{itemize}

We analyze these relationships under the partial order of ``coarsens to'' on the set of decompositions, establishing a hierarchy of coarseness. These relationships lead to a comprehensive classification system, which we call the \textbf{Finest Decompositions' Architecture (FDA)}. The FDA framework classifies all reaction networks into six distinct classes, primarily separating networks with \textbf{Independent Linkage Classes (ILC)}—where the visual structure aligns with algebraic independence—from those with \textbf{Dependent Linkage Classes (DLC)}, where visual connectivity hides underlying dependencies. Both Independent Linkage Classes (ILC) and Dependent Linkage Classes (DLC) are each categorized into three distinct subclasses, namely ILC-1, ILC-2, ILC-3 and DLC-1, DLC-2, DLC-3. These six categories constitute the complete set of Finest Decompositions' Architectures (FDA) for a reaction network.

To better characterize the FDA classes, we introduce two new network properties:
\begin{itemize}
    \item[1.] \textbf{Deficiency Difference ($\Delta$):} For any decomposition $D$, we set $\Delta_D = |\delta - \delta_D|$, where $\delta$ is the network deficiency and $\delta_D$ is the sum of the subnetwork deficiencies. The deficiency difference is expressed as the vector $(\Delta_{LCD}, \Delta_{FID}, \Delta_{FIID})$.
    \item[2.] \textbf{Common Complexes Cardinality ($|\mathscr{CC}|$):} This represents the number of common complexes in the FID ($|\mathscr{CC}| := CC_{FID}$), which plays a decisive role in identifying and characterizing DLC architectures.
\end{itemize}

As a main result, we show that the range of values of $\Delta$ uniquely determines the ILC classes and DLC-3. Furthermore, while DLC-1 and DLC-2 networks share the same $\Delta$ value range, they differ in their $|\mathscr{CC}|$-value, being zero for DLC-1 and positive for DLC-2. This characterization provides a convenient way of computing a network’s FDA class. We provide numerous examples and identify interesting subclasses within these architectures.

Another main result demonstrates that deficiency zero networks form a proper subset of the \textbf{FDE (Finest Decompositions Equality)} subclass of ILC-1, consisting of networks where $FID = FIID$ (or equivalently $\Delta = (0, 0, 0)$). This leads to a new characterization of deficiency zero networks. Additionally, we observe that many biochemical models belong to the \textbf{FDC (Finest Decomposition Coarsening)} subset—comprised of ILC-1, DLC-1, and DLC-2—which are those networks where the FIID is a refinement of (or coarsens to) the FID.

Because fundamental results in Chemical Reaction Network Theory rely on the ILC property, the FDA framework provides an essential bridge between a network's visual structure and its kinetic attributes. In particular, the Deficiency One Theorem and various power-law system analyses rely on the ILC property. Furthermore, in mass action systems with the ILC property, the existence of steady states for the entire network is equivalent to the existence of steady states within its individual linkage-class subnetworks.

The remainder of this paper is structured systematically to explore the FDA. Section \ref{fundamentals:CRN} establishes the mathematical foundation by reviewing the fundamentals of chemical reaction networks and kinetic systems. Section \ref{section:FDA} introduces the FDA framework, defining the six architecture classes based on the coarsening relationships between the network's three finest decompositions: LCD, FID, and FIID.
Section \ref{section:FDA:ILC} discusses the finest decompositions' architecture subclasses of ILC networks. On the other hand, Section \ref{kinetic properties:ILC} provides kinetic properties of ILC systems.
Section \ref{section:FDA:DLC} discusses FDA classes of DLC networks. Finally, Section \ref{summary:conclusion:recommendation} provides summary, conclusion, and recommendations.

\section{Fundamentals of reaction networks and kinetic systems}
\label{fundamentals:CRN}

In this section, we discuss fundamental concepts in reaction networks and kinetic systems, drawing heavily from the book \cite{feinberg:book} and lectures \cite{feinberg:lecture} of M. Feinberg.

\subsection{Chemical reaction networks}

\begin{definition}
A {\underline{\emph{chemical reaction network}}} 
is a triple $(\mathcal{S}, \mathcal{C}, \mathcal{R})$ of nonempty finite sets where
    \begin{itemize}
        \item[a.] $\mathcal{S} = \left\{A_i:i=1,2,\ldots,m \right\}$ is the set of {\underline{\emph{species}}},
        \item[b.] $\mathcal{C} = \{C_i:i=1,2,\ldots,n\}$ is the set of {\underline{\emph{complexes}}}, which are non-negative linear combinations of the species, and
        \item[c.] $\mathcal{R} = \{R_i:i=1,2,\ldots,r\} \subset \mathcal{C} \times \mathcal{C}$ is the set of {\underline{\emph{reactions}}}
    \end{itemize}
such that
    \begin{itemize}
        \item[i.] $\left( {{C_i},{C_i}} \right) \notin \mathcal{R}$ for each $C_i \in \mathcal{C}$, and
        \item[ii.] for each $C_i \in \mathcal{C}$, there exists $C_j \in \mathcal{C}$ such that $\left( {{C_i},{C_j}} \right) \in \mathcal{R}$ or $\left( {{C_j},{C_i}} \right) \in \mathcal{R}$.
    \end{itemize}
\end{definition}

We typically denote a reaction 
$(y,y')$ as $y \to y'$, 
where $y$ is the {\underline{\emph{reactant}} \underline{\emph{complex}}} while $y'$ is the {\underline{\emph{product complex}}}. The associated {\underline{\emph{reaction vector}}} is defined by the difference $y'-y$. The linear subspace $S$ of $\mathbb{R}^m$ is spanned by all the reaction vectors in the network and is known as the {\underline{\emph{stoichiometric subspace}}} of the network, i.e., $S = \mathrm{span}\{y' - y \in \mathbb{R}^m \mid y \rightarrow y' \in \mathcal{R}\}$.

\begin{definition}
The {\underline{\emph{molecularity matrix}}}, denoted by $Y$, is an $m\times n$ matrix such that $Y_{ij}$ is the stoichiometric coefficient of species $A_i$ in complex $C_j$.
The {\underline{\emph{incidence matrix}}}, denoted by $I_a$, is an $n\times r$ matrix such that 
$${\left( {{I_a}} \right)_{ij}} = \left\{ \begin{array}{rl}
 - 1&{\rm{ if \ }}{C_i}{\rm{ \ is \ the\ reactant \ complex \ of \ reaction \ }}{R_j},\\
 1&{\rm{  if \ }}{C_i}{\rm{ \ is \ the\ product \ complex \ of \ reaction \ }}{R_j},\\
0&{\rm{    otherwise}}.
\end{array} \right.$$
The {\underline{\emph{stoichiometric matrix}}}, denoted by $N$, is the $m\times r$ matrix 
$N=YI_a$.
\end{definition}

Alternatively, the stoichiometric matrix of a CRN has columns representing the coefficients of species in each corresponding reaction vector. Thus, the dimension of the stoichiometric subspace of a network is precisely the rank of its stoichiometric matrix.

Complexes can be treated as vertices and reactions as edges, viewing CRNs as directed graphs. If there is a path between two complexes $C_i$ and $C_j$, they are said to be \underline{\emph{connected}}. If there is a directed path from $C_i$ to $C_j$ and vice versa, they are said to be \underline{\emph{strongly connected}}. A subgraph where any two vertices are (strongly) connected forms a \underline{\emph{(strongly) connected component}}. These (strong) connected components precisely correspond to the \underline{\emph{(strong) linkage classes}} of a CRN.
A maximal strongly connected subgraph where there are no edges from a complex in the subgraph to a complex outside the subgraph is said to be a \underline{\emph{terminal strong linkage classes}}.
Let $\ell$ denote the number of linkage classes and $sl$ the number of strong linkage classes. The CRN is said to be \underline{\emph{weakly reversible}} if $sl=\ell$. Alternatively, a CRN is weakly reversible if each of its reactions belongs to a directed cycle.
Furthermore, let $n_r$ be the number of reactant complexes of a CRN. A CRN is {\emph{cycle terminal}} if and only if $n = n_r$.

\begin{definition}
The \underline{\emph{deficiency}} of a CRN is given by the formula $\delta:=n-\ell-s$ where $n$ is the number of complexes, $\ell$ is the number of linkage classes, and $s:={\rm{dim \ }} S$.
\end{definition}

Alternative formulas for the deficiency of a network where each linkage class contains only one terminal strong linkage class from Proposition 2.10 of \cite{boros:thesis} are given by
\begin{equation}
\label{alternative:definition}
\delta = \dim \ker S - \dim \ker I_a= \dim(\ker Y \cap {\rm Im \ } I_a).   
\end{equation}

\subsection{Kinetic systems}

\begin{definition}
A \underline{\emph{kinetics}} for a reaction network $\mathcal{N}=(\mathcal{S}, \mathcal{C}, \mathcal{R})$ is an assignment to
each reaction $y \to y' \in \mathcal{R}$ of a continuously differentiable {rate function} $K_{y\to y'}: \mathbb{R}^m_{{\geq} 0} \to \mathbb{R}_{\ge 0}$ such that this positivity condition holds:
$K_{y\to y'}(c) > 0$ if and only if ${\rm{supp \ }} y \subset {\rm{supp \ }} c$, where ${\rm{supp \ }} y$ refers to the support of the vector $y$ (the set of species with nonzero coefficient in $y$).
The pair $\left(\mathcal{N},K\right)$ is called a \underline{\emph{(chemical) kinetic system}}.
\end{definition}

 \begin{definition}
	A \underline{\emph{power law kinetics}} is a kinetics of the form $${K_i}\left( x \right) = {k_i}\prod\limits_j {{x_j}^{{F_{ij}}}}  := {k_i}{{x^{{F_{i}}}}} $$ for each reaction $i =1,\ldots,r$ where ${k_i} \in {\mathbb{R}_{ > 0}}$ and ${F_{ij}} \in {\mathbb{R}}$. The $r \times m$ matrix $F=\left[ F_{ij} \right]$ is called the \underline{\emph{kinetic order matrix}}, and $k \in \mathbb{R}^r$ is called the \underline{\emph{rate vector}}.
	\label{def:power:law}
\end{definition}

If the kinetic order matrix is the transpose of the molecularity matrix, then the kinetics is \underline{\emph{mass action kinetics (MAK)}}.

\begin{definition}
A power law system has \underline{\emph{reactant-determined kinetics}} (i.e., of type PL-RDK) if for any two reactions indexed by $i$ and $j$ with identical reactant complexes, the corresponding rows of kinetic orders in $F$ are also identical, i.e., ${F_{ik}} = {F_{jk}}$ for $k = 1,2,...,m$. It has \underline{\emph{non-reactant-determined kinetics}} (i.e., of type PL-NDK) if there exist two reactions with the same reactant complexes whose corresponding rows in $F$ are not identical.
\end{definition}

Furthermore, due to S. M\"uller and G. Regensburger \cite{MURE2014}, we present the following definitions:

\begin{definition}
	The $m \times n$ matrix $\widetilde{Y}$ is defined as
	$${\left( {\widetilde{Y}} \right)_{ij}} = \left\{ \begin{array}{l}
		{F_{ki}}{\text{  \ \ \ \ \ \ \  if \ $j$ {\text{is the reactant complex in reaction }}$k$}}\\
		{0} {\text{  \ \ \ \ \ \ \  \ \   otherwise \ }}
	\end{array} \right..$$
\end{definition}

\begin{definition}
	The $m \times n_r$ T-matrix is the truncated $\widetilde{Y}$ where the non-reactant columns are deleted.
\end{definition}

\begin{definition}
	The block matrix $\widehat{T} \in \mathbb{R}^{(m+l) \times n_r}$ is defined as
	$\widehat{T} = \left[ {\begin{array}{*{20}{c}}
			T\\
			{{L^\top}}
	\end{array}} \right]$
where the $n_r \times \ell$ matrix $L = [e_1, e_2, \ldots , e_\ell]$, and $e_i$ is a characteristic vector for the linkage class $\mathcal{L}_i$.
\end{definition}

	The \underline{\emph{species formation rate function}} of $(\mathcal{N},K)$ is given by $$f\left( x \right) = \displaystyle \sum\limits_{{y} \to {y'} \in \mathcal{R}} {{K_{{y} \to {y'}}}\left( x \right)\left( {{y'} - {y}} \right)}$$ where $x$ is a vector of concentrations of the species over time.
It can also be expressed as
$f(x) = NK(x)$ where $N$ is the stoichiometric matrix and $K$ is the vector of rate functions of $\mathcal{N}$.
The system of \underline{\emph{ordinary differential equations}} (ODEs) of the kinetic system is given by $\dfrac{{dx}}{{dt}} = f\left( x \right)$.

An \underline{\emph{equilibrium}} or a \underline{\emph{steady state}} is a vector 
$c$ of species concentrations where $f(c)=0$. The \underline{\emph{set of positive equilibria}} of a kinetic system $\left(\mathcal{N},K\right)$ is 
\begin{equation*}
    E_+:={E_ + }\left(\mathcal{N},K\right)= \left\{ {x \in \mathbb{R}^m_{>0}|f\left( x \right) = 0} \right\}.
\end{equation*}

A CRN is said to admit \underline{\emph{multiple equilibria}} or is \underline{\emph{multistationary}} if there exist positive rate constants such that the ODE system admits more than one positive equilibrium within a stoichiometric class.

The \underline{\emph{set of complex balanced equilibria}} \cite{HornJackson} is given by 
	\[Z_+:={Z_ + }\left(\mathcal{N},K\right) = \left\{ {x \in \mathbb{R}_{ > 0}^m|{I_a}  K\left( x \right) = 0} \right\} \subseteq {E_ + }.\]
Hence, a positive vector $c \in \mathbb{R}^m$ is \underline{\emph{complex balanced}} if $K\left( c \right)$ is contained in the kernel of the incidence matrix, i.e., ${\text{Ker }}{I_a}$. A kinetic system is \underline{\emph{complex balanced}} if it has a complex balanced equilibrium.

\subsection{Decomposition theory of reaction networks}
A CRN can be decomposed into smaller networks called \emph{subnetworks} by partitioning its set of reactions:

\begin{definition}
A \underline{\emph{decomposition}} of $\mathcal{N}$ is a set of subnetworks $\{\mathcal{N}_1, \mathcal{N}_2,...,\mathcal{N}_k\}$ of $\mathcal{N}$ induced by a partition $\{\mathcal{R}_1, \mathcal{R}_2,...,\mathcal{R}_k\}$ of its reaction set $\mathcal{R}$. 
\end{definition}

\begin{definition}
A network decomposition $\mathcal{N} = \mathcal{N}_1 \cup \mathcal{N}_2 \cup \ldots \cup \mathcal{N}_k$  is a \underline{\emph{refinement}} of
$\mathcal{N} = {\mathcal{N}'}_1 \cup {\mathcal{N}'}_2 \cup \ldots \cup {\mathcal{N}'}_{k'}$ 
(and the latter a \underline{\emph{coarsening}} of the former) if it is induced by a refinement  
$\{\mathcal{R}_1, \mathcal{R}_2,\ldots,\mathcal{R}_k\}$
of $\{{\mathcal{R}'}_1, {\mathcal{R}'}_2, \ldots, {\mathcal{R}'}_{k'}\}$, i.e., 
each ${\mathcal{R}}_i$ is contained in an ${\mathcal{R}'}_j$ for some $j\in \left\{1,2,\ldots,k'\right\}$. 
\end{definition}

\begin{definition}
A decomposition $\mathcal{N} = \mathcal{N}_1 \cup \mathcal{N}_2 \cup \ldots \cup \mathcal{N}_k$ with $\mathcal{N}_i = (\mathcal{S}_i,\mathcal{C}_i,\mathcal{R}_i)$ for $i=1,2,\ldots,k$ is a \underline{\emph{$\mathcal{C}$-decomposition}} if for each pair of distinct $i$ and $j$, $\mathcal{C}_i$ and $\mathcal{C}_j$ are disjoint.
\end{definition}

\begin{definition}
A network decomposition 
is \underline{\emph{independent}} if the network stoichiometric subspace is a direct sum of the subnetwork stoichiometric subspaces. A decomposition is \underline{\emph{incidence-independent}} if the network incidence matrix is a direct sum of the subnetwork incidence matrices. The decomposition is \underline{\emph{bi-independent}} if it is both independent and incidence-independent.
\end{definition}

\begin{remark}
    An equivalent formulation for an incidence-independent decomposition is the following equation \cite{fari2021}:
    \begin{equation}
    \label{incidence:check}
      n-\ell = \displaystyle\sum_{i=1}^k (n_i - \ell_i).  
    \end{equation}
    On the other hand, an equivalent formulation for an independent decomposition is the following equation \cite{fari2021}:
    \begin{equation}
        s=\displaystyle\sum_{i=1}^ks_i.
    \end{equation}
\end{remark}

Lemma 1 of \cite{fortun2} shows that for an independent decomposition,
\begin{equation}
\label{independent:le}
    \delta \le \delta_1 +\delta_2 + \cdots +\delta_k.
\end{equation}
Analogously, Proposition 6 of \cite{fari2021} states that for an incidence-independent decomposition, 
\begin{equation}
\label{independent:ge}
    \delta \ge \delta_1 +\delta_2 + \cdots +\delta_k
\end{equation}
and that $\mathcal{C}$-decompositions form a subset of all incidence-independent decompositions. Hence, independent linkage classes form the primary example of a bi-independent decomposition.
Furthermore, from Equations \ref{independent:le} and \ref{independent:ge}, for a bi-independent decomposition,
\begin{equation}
    \label{independent:eq}
    \delta = \delta_1 +\delta_2 + \cdots +\delta_k.
\end{equation}

\begin{lemma} \cite{Hernandez:WRCF}
\label{lemma:two}
For any network decomposition $\mathcal{N}=\mathcal{N}_1\cup\mathcal{N}_2\cup \ldots \cup\mathcal{N}_k$, the following statements are equivalent:
\begin{enumerate}
    \item independent and $\delta=\delta_1+\delta_2+\ldots+\delta_k$,
    \item incidence independent and $\delta=\delta_1+\delta_2+\ldots+\delta_k$,
    \item bi-independent.
\end{enumerate}
\end{lemma}

Importantly, M. Feinberg established the fundamental relationship between an independent decomposition and the set of positive equilibria:

\begin{theorem}[Remark 5.4, \cite{feinberg12}]\label{feinberg decomp} Let $(\mathcal N, K)$ be a kinetic system underlying network decomposed into $\mathcal N=\mathcal N_1\cup \mathcal N_2 \cup \ldots \cup \mathcal N_k$ and $E_+(\mathcal N_i, K_i)=\{x\in \mathbb R^{m}_{>0}\;|\;N_iK_i(x)=0\}$, then 
\begin{enumerate}
    \item[i.] $\displaystyle \bigcap_{i=1}^k E_+(\mathcal N_i, K_i)\subseteq E_+(\mathcal N, K)$ and
    \item[ii.] the equality holds if the decomposition is independent.
\end{enumerate}
\end{theorem}

The analogue of Feinberg's decomposition result for incidence-independent decompositions and complex-balanced equilibria is given as follows:

\begin{theorem}[Theorem 4, \cite{fari2021}]\label{farinas decomp}
Let $\mathcal N=\mathcal N_1\cup \cdots \cup \mathcal N_k$ be a decomposition. Let $K$ be any kinetics, and $Z_+(\mathcal N, K)$ and $Z_+(\mathcal N_i, K_i)$ be the set of complex balanced equilibria of $\mathcal N$ and $\mathcal N_i$, respectively. Then, 
\begin{equation*}
     \displaystyle \bigcap_{i=1}^k Z_+(\mathcal N_i, K_i)\subseteq Z_+(\mathcal N, K).
\end{equation*}
If the decomposition is incidence-independent, then
\begin{enumerate}
    \item[i.] $\displaystyle \bigcap_{i=1}^k Z_+(\mathcal N_i, K_i)=Z_+(\mathcal N, K)$ and
    \item[ii.] $Z_+(\mathcal N, K)\neq \varnothing$ implies that $Z_+(\mathcal N_i, K_i)\neq \varnothing$ for $i=1, \cdots, k$. 
\end{enumerate}
\end{theorem}

\section{Overview of the finest decompositions' architecture (FDA)}
\label{section:FDA}

In this section, we define the finest decompositions' architecture (FDA) of a reaction network, or more precisely its FDA class. We first review concepts and results concerning a network's three finest decompositions and then define the FDA classes based on coarsening relationships between them. After a formal verification that each reaction network is assigned a unique FDA class, we discuss some algebraic aspects of our results in the context of the semigroup of subnetworks of the reaction network.

\subsection{The three finest decompositions of a CRN}

\subsubsection{Review of FID and FIID concepts and key properties}

We begin with a simple example to review independent decompositions of chemical reaction networks (CRNs), and illustrate FID and FIID in the latter part of this section. 

\begin{example}
\label{example:Anderies}
    Consider the CRN of Anderies et al. \cite{anderies2013topology,fortun2018deficiency} pre-industrial model that describes the Earth's carbon cycle among the three carbon pools: land ($A_1$), atmosphere ($A_2$), and ocean ($A_3$):
    \begin{align*}
    A_1+2A_2 &{\stackrel{R_1}\rightarrow} A_2+2A_1\\
        A_1+A_2 &{\stackrel{R_2}\rightarrow} 2A_2\\
        A_2 &\underset{R_4}{\stackrel{R_3}\rightleftharpoons} A_3. 
    \end{align*}
Suppose we partition the reaction set $\mathcal{R}=\{R_1,R_2,R_3,R_4\}$ into subsets $\{R_1,R_2\}$ and $\{R_3,R_4\}$. This induces two subnetworks $\mathcal{N}_1$ and $\mathcal{N}_2$, each of which has two reactions.

The stoichiometric matrix of the network is
    \begin{center}
    $N=\begin{blockarray}{ccccc}
        \matindex{$R_1$} & \matindex{$R_2$} & \matindex{$R_3$} & \matindex{$R_4$} \\
        \begin{block}{[cccc]c}
        1 & -1 & 0 & 0  & \matindex{$A_1$}\\
        -1 & 1 & -1 & 1 & \matindex{$A_2$}\\
        0 & 0 & 1 & -1 & \matindex{$A_3$}\\
        \end{block}
        \end{blockarray}$,
    \end{center}
while the stoichiometric matrices of the two subnetworks ($\NN_1$ and $\NN_2$) are
    $$N_1=\begin{blockarray}{ccc}
        \matindex{$R_1$} & \matindex{$R_2$}\\
        \begin{block}{[cc]c}
        1 & -1 & \matindex{$A_1$}\\
        -1 & 1 & \matindex{$A_2$}\\
        0 & 0 & \matindex{$A_3$}\\
        \end{block}
        \end{blockarray}
    {\text{and }}
    N_2=\begin{blockarray}{ccc}
         \matindex{$R_3$} & \matindex{$R_4$} \\
        \begin{block}{[cc]c}
         0 & 0 & \matindex{$A_1$}\\
         -1 & 1 & \matindex{$A_2$}\\
         1 & -1 & \matindex{$A_3$}\\
        \end{block}
        \end{blockarray}.$$
    Since rank $N=2$, rank $N_1 = 1$ and rank $N_2 = 1$, then the sum of the ranks of the stoichiometric matrices of the subnetworks is the rank of the stoichiometric matrix of the whole network, i.e., rank $N$ = rank $N_1$ $+$ rank $N_2$. Thus, the decomposition is independent.
\end{example}

\begin{theorem} \cite{had2022}
\label{FID:unique}
    Let $R$ be a finite set of vectors. $R$ has exactly one independent (incidence-independent) decomposition with no independent (incidence-independent) refinement.
\end{theorem}

Theorem \ref{FID:unique} deals with vectors in general, which can be put in the context of reaction vectors of a CRN, i.e., columns of the stoichiometric matrix. It establishes the existence and uniqueness of a network decomposition that cannot be further decomposed into another independent decomposition that we define as follows: 

\begin{definition} 
    The unique independent (incidence-independent) decomposition of a CRN with no independent refinement is called its {\emph{finest independent decomposition}} or {\emph{FID}} {(\emph{finest incidence-independent decomposition}} or {\emph{FIID}}). 
\end{definition}

\begin{remark}
    The uniqueness of FID and FIID for each CRN indicates that each is a new network property.
\end{remark}

\begin{remark}
We built a MATLAB tool called DECENT (DEcompositions of Chemical rEaction NeTworks) based on the earlier results of Hernandez et al. \cite{had2022,hernandez:delacruz1}, to compute the subnetworks under the FIID. This tool also incorporates the computation of subnetworks under the FID, which was previously developed in \cite{INDECS}. Furthermore, the tool is able to display the network deficiencies of a given CRN and its subnetworks under the FID and FIID.
This tool is stored in GitHub \cite{DECENT}. We provide a brief discussion in the Appendix \ref{DECENT:Anderies}.
\end{remark}

\begin{example}
\label{example:Anderies2}
    Reconsider the Anderies et al. \cite{anderies2013topology,fortun2018deficiency} pre-industrial network:
    \begin{align*}
    A_1+2A_2 &{\stackrel{R_1}\rightarrow} A_2+2A_1\\
        A_1+A_2 &{\stackrel{R_2}\rightarrow} 2A_2\\
        A_2 &\underset{R_4}{\stackrel{R_3}\rightleftharpoons} A_3. 
    \end{align*}

    Applying DECENT, we obtain the FID consisting of $\{R_1,R_2\}$ and $\{R_3,R_4\}$, while we obtain the FIID consisting of $\{R_1\}$, $\{R_2\}$ and $\{R_3,R_4\}$. 
\end{example}

\begin{example}\label{example:Anderies2.1}
    Consider a translated version \cite{HHLL2022,Johnston2014} of the Anderies et al. pre-industrial network in Example \ref{example:Anderies}:
    \begin{align*}
         &A_1 \underset{R_2}{\stackrel{R_1}\rightleftharpoons} A_2 \underset{R_4}{\stackrel{R_3}\rightleftharpoons} A_3. 
    \end{align*}
    Both the FID and FIID consist of $\{R_1,R_2\}$ and $\{R_3,R_4\}$. 
\end{example}

We recall the following useful lemmas from previous work:

\begin{lemma} \cite{fari2021}
\label{lemma:one}
  For a reaction network $\mathcal{N}$, any coarsening of an independent (incidence independent) decomposition is independent (incidence independent).  
\end{lemma}

The following result highlights the importance of FID (FIID) as a ``generator'' since any independent (incidence independent) decomposition can be constructed from coarsening the FID (FIID):

\begin{lemma} \cite{had2022}
\label{lemma:three}
Any independent (incidence independent) decomposition is a coarsening of the
FID (FIID).
\end{lemma}

Independent (incidence independent) decompositions define ``building blocks'' of
positive (complex balanced) equilibria of a kinetic system.

We observed that in many reaction networks of biochemical and environmental systems that the FID is a coarsening of the FIID. We call such a network an \emph{FDC (Finest Decomposition Coarsening)} network.

Our first main result is the following:

\begin{theorem} (FDC Characterization).
For a reaction network $\mathcal{N}$, the following statements are
equivalent.
\label{theorem:FDC:equiv}
\begin{enumerate}
    \item[i.] $\mathcal{N}$ is an FDC network.
    \item[ii.] The FID of $\mathcal{N}$ is bi-independent.
    \item[iii.] Any independent decomposition of $\mathcal{N}$ is bi-independent.
    \item[iv.] The deficiency of $\mathcal{N}$ is equal to the sum of the deficiencies of the subnetworks of any independent decomposition, i.e., $\delta=\delta_1+\delta_2+\ldots+\delta_k$.
    \label{theorem:characterization}
\end{enumerate}
\end{theorem}

\begin{proof}
i) $\implies$ ii): Since the FID by assumption is a coarsening of the FIID, it follows from Lemma \ref{lemma:one} that it
is also incidence independent, hence bi-independent. ii) $\implies$ iii): By Lemma \ref{lemma:three}, the independent
decomposition is a coarsening of the FID. Hence, by Lemma \ref{lemma:one} and the assumption, it is also
incidence independent, hence bi-independent. iii) $\implies$ iv): This follows from Lemma \ref{lemma:two}. iv)
$\implies$ i): By assumption and Lemma \ref{lemma:two}, the FID is incidence independent. 
\end{proof}

\subsubsection{Discussion of the LCD as the {coarsest} $\mathscr{C}$-decomposition}
\label{LCD:C:decomposition}

Fari{\~n}as et al. \cite{fari2021} introduced the concepts of $\mathscr{C}$-decomposition and $\mathscr{C}^*$-decomposition (defined in Definition \ref{definition:C:Cstar:decomposition}), demonstrating that both are incidence-independent. This was established prior to the formal definition of common complexes. Subsequently, Fontanil and Mendoza \cite{fontanil} introduced the general ``set $\mathscr{CC}_D$ of common complexes of a decomposition $D$,''
consisting of all complexes contained in sets of complexes of at least two distinct subnetworks under the decomposition $D$. Hence,
$$\mathscr{CC}_D = \bigcup_{i<j} (\mathcal{C}_i \cap \mathcal{C}_j)$$where $i = 1, 2, \ldots, k-1$ and $k$ represents the number of subnetworks \cite{fontanil}.
In this paper, the new concept of ``common complexes of a network $\mathscr{CC}$'' is introduced
and defined as $\mathscr{CC}:=\mathscr{CC}_{FID}$. Its cardinality is called $\mathscr{CC}$-cardinality. In particular, the following result holds \cite{fari2021}:

\begin{proposition}
    {Any decomposition with $|\mathscr{C}\mathscr{C}_D|\le1$ is incidence-independent.}
    \label{cdecomposition:incidenceindependent}
\end{proposition}

\begin{definition}
A decomposition with $|\mathscr{C}\mathscr{C}_D|$ equals 0 (equals $1$) is also called a $\mathscr{C}$-decomposition ($\mathscr{C}^*$-decomposition).
\label{definition:C:Cstar:decomposition}
\end{definition}

\begin{remark}
Together with Proposition \ref{cdecomposition:incidenceindependent}, the structure theorem for $\mathscr{C}$-decompositions of Fari{\~n}as et al. \cite{fari2021} as follows reveals that the linkage class decomposition (LCD) is the {coarsest} $\mathscr{C}$-decomposition of a reaction network.
\label{LCD:coarsening:FIID}
\end{remark}

\begin{theorem} (Structure Theorem for $\mathcal{C}$-decomposition)
Let $\mathcal{L}_1$, $\mathcal{L}_2$, $\ldots$, $\mathcal{L}_\ell$ be the linkage classes of a network $\mathcal{N}$. A decomposition $\mathcal{N}=\mathcal{N}_1 \cup \mathcal{N}_2 \cup \cdots \cup \mathcal{N}_k$ is a $\mathcal{C}$-decomposition if and only if each $\mathcal{N}_i$ is the union of linkage classes and each linkage class is contained in only one $\mathcal{N}_i$. In other words, the linkage class decomposition is a refinement of the specified decomposition $\mathcal{N}=\mathcal{N}_1 \cup \mathcal{N}_2 \cup \cdots \cup \mathcal{N}_k$.
\label{structure:theorem:farinas}
\end{theorem}

Thus, the LCD is the fundamental connection between the topological and algebraic structure of a network. In this context, weakly reversible networks, where, due to the coincidence of linkage and strong linkage classes, all relevant connectivity information is contained in the former, represent a high level of alignment between topology and algebra in a network.

\subsection{The definition of the FDA of a CRN}

We now use the coarsening relationships between FID, FIID, and LCD to define the six FDA classes of reaction networks. A key property of this classification is that it separates ILC networks, i.e., those with independent linkage classes, and DLC networks, i.e., those with dependent linkage classes, into three subsets each. In the following definition, for decompositions $A$ and $B$, $A\leq B$ denotes $B$ coarsens to $A$ and $A<B$ means $B$ coarsens to $A$ but is not equal to $A$.

\begin{definition}
    The finest decompositions' architecture (FDA) class of a reaction network is one of the following six classes:
    \begin{enumerate}
        \item {\bf ILC-1}: $\text{LCD}\leq\text{FID}\leq\text{FIID}$
        \item {\bf ILC-2}: $\text{LCD}\leq\text{FIID}<\text{FID}$
        \item {\bf ILC-3}: $\text{LCD}<\text{FID}$, $\text{LCD}<\text{FIID}$, FID and FIID are not ordered
        \item {\bf DLC-1}: $\text{FID}<\text{LCD}\leq\text{FIID}$
        \item {\bf DLC-2}: $\text{FID}<\text{FIID}$, $\text{LCD}<\text{FIID}$, FID and LCD are not ordered
        \item {\bf DLC-3}: $\text{LCD}\leq\text{FIID}$, FID and LCD are not ordered, FID and FIID are not ordered.
    \end{enumerate}
\end{definition}

The FDA classes are graphically depicted in Figure 1, whereby an arrow from $A$ to $B$ indicates $A$ coarsens to $B$ and an $\ne$ beside the arrow indicates the inequality of $A$ and $B$.

\begin{figure}
\centering
    \begin{center}
    \includegraphics[width=13cm,height=8cm,keepaspectratio]{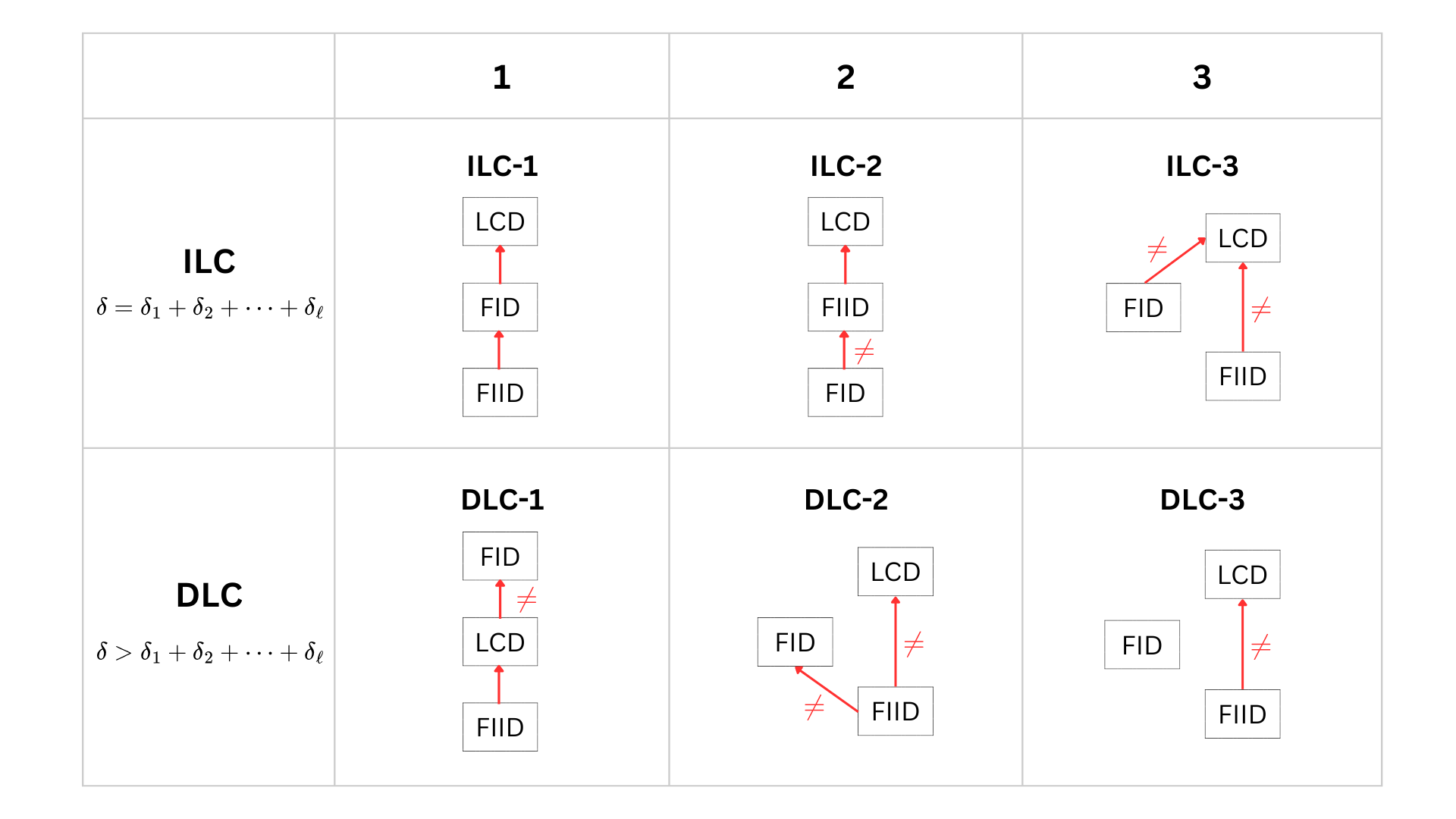}
    \caption{\textbf{The finest decompositions' architecture of a chemical reaction network.} The FDA hierarchy is established by the partial order 'coarsens to'. While ILC networks are characterized by deficiency additivity over their linkage classes, DLC networks lack this alignment. The solid arrow $\rightarrow$ indicates that the decomposition at the tail is a refinement of (or ''coarsens to'') the decomposition at the head. The inequality sign $\ne$ marks where a coarsening relationship exists but the two decompositions are not equal.}\label{fig:FDAtypes}
    \vspace{-1.5em}
    \end{center}    
\end{figure}

We now show that each reaction network is assigned a unique FDA class by exhausting the relations of the three decompositions:
\begin{enumerate}
    \item[D1.] $\text{LCD}=\text{FID},\text{ FID}=\text{FIID},\text{ LCD}=\text{FIID}$
    \item[D2.] $\text{LCD}=\text{FID},\text{ FID}\ne\text{FIID},\text{ LCD}\ne\text{FIID}$
    \item[D3.] $\text{LCD}\ne\text{FID},\text{ FID}=\text{FIID},\text{ LCD}\ne\text{FIID}$
    \item[D4.] $\text{LCD}\ne\text{FID},\text{ FID}\ne\text{FIID},\text{ LCD}=\text{FIID}$
    \item[D5.] {$\text{LCD}\ne\text{FID},\text{ FID}=\text{FIID},\text{ LCD}=\text{FIID}$}
    \item[D6.] {$\text{LCD}=\text{FID},\text{ FID}\ne\text{FIID},\text{ LCD}=\text{FIID}$} 
    \item[D7.] {$\text{LCD}=\text{FID},\text{ FID}=\text{FIID},\text{ LCD}\ne\text{FIID}$}
    \item[D8.] $\text{LCD}\ne\text{FID},\text{ FID}\ne\text{FIID},\text{ LCD}\ne\text{FIID}$ 
\end{enumerate}

Inspecting the equalities and inequalities, cases 5, 6 and 7 are impossible to happen. Table \ref{possible} gives all the possibilities for each FDA class.

\begin{table}[h!]
\centering
\begin{tabular}{|c|c|c|c|c|c|c|}
\hline
\textbf{} & \textbf{ILC-1} & \textbf{ILC-2} & \textbf{ILC-3} & \textbf{DLC-1} & \textbf{DLC-2} & \textbf{DLC-3} \\ \hline
D1 & \checkmark & x & x & x & x & x \\ \hline
D2 & \checkmark & x & x & x & x & x \\ \hline
D3 & \checkmark & x & x & x & \checkmark & x \\ \hline
D4 & x & \checkmark & \checkmark & \checkmark & \checkmark & \checkmark \\ \hline
D8 & \checkmark & \checkmark & \checkmark & \checkmark & \checkmark & \checkmark \\ \hline
\end{tabular}
\caption{Possible decomposition relations for each FDA}
\label{possible}
\end{table}

\subsection{Algebraic aspects of CRNs}

The set of all subnetworks of a reaction network, together with the union operation forms a commutative idempotent semigroup, in fact a monoid, if one allows the empty set as a network. If one uses multiplicative notation, a network decomposition corresponds to ``factoring'' a network. Furthermore, coarsening defines a natural partial ordering on the set of decompositions of a network.

In this view, independent decompositions and incidence-independent decompositions are two different ways of determining ``relatively prime'' factors. Both approaches are simple, allowing only single factors due to the operation's idempotency. Both ways are also preserved by partial ordering.
Consequently, the FID and FIID provide two different ways of unique prime factorization. FDC networks hence constitute the subset where these prime factorizations are aligned by the partial ordering, which constitutes an appealing algebraic architecture. 
An appealing algebraic property is the uniqueness of this prime factorization, which occurs in the
four classes ILC-1, DLC-1, DLC-2, and ILC-3.
The first three classes constitute the set of FDC networks which
are characterized by Theorem \ref{theorem:characterization}. The observation of the FDC
property in numerous biochemical networks (s. Sections \ref{section:FDA:ILC} and \ref{section:FDA:DLC}) was the initial motivation for
this study.

\section{The finest decompositions' architecture classes of independent linkage class (ILC) networks}
\label{section:FDA:ILC}

In this section, we introduce the deficiency difference $\Delta$ of a reaction network and show that its range of values characterizes the FDA classes of ILC networks. We provide examples for each class and identify some interesting subclasses. An interesting result is that all deficiency zero networks are contained in the FDE subclass of ILC-1 as it leads to a new characterization of such networks.

\subsection{Characterization and examples of ILC-1 networks}

Recall that the deficiency $\delta_D$ of a decomposition $D$ is the sum of the deficiencies of its subnetworks. We introduce the following new concepts:

\begin{definition}
    The \textit{deficiency difference of a decomposition} is given by $\Delta_D=|\delta-\delta_D|$ where $\delta$ is the network deficiency. Its range, $r(\Delta_D)$, is an element of $\{=,>,\geq\}$ according to whether $\Delta_D$ is zero, positive, or nonnegative, respectively.
\end{definition}

A direct consequence of Lemma 1 of \cite{fortun2}, Proposition 6 of \cite{fari2021}, and Lemma \ref{lemma:two} is the following:
\begin{proposition}\label{prop:defdiff}
    Let $D$ be a decomposition of a network. If $D$ is independent, then $\Delta_D=\delta_D-\delta$, while if $D$ is incidence independent, then $\Delta_D=\delta-\delta_D$. Furthermore, $D$ is bi-independent if and only if $\Delta_D=0$.
\end{proposition}

\begin{definition}
    The \textit{deficiency difference of a reaction network} is the vector $$\Delta=(\Delta_{\text{LCD}},\Delta_{\text{FID}},\Delta_{\text{FIID}}).$$ Its range $r(\Delta)$ is the vector of ranges of its components.
\end{definition}

We can now characterize the ILC-1 class and its subclass FDE:
\begin{proposition}
    Let $\mathcal N$ be a network with deficiency difference $\Delta$.
    \begin{enumerate}
        \item[i.] $\mathcal N$ is an ILC-1 network if and only if $r(\Delta)=(=,=,\geq)$.
        \item[ii.] $\mathcal N$ is an FDE network if and only if $r(\Delta)=(=,=,=)$ or simply $\Delta=(0,0,0)$.
    \end{enumerate}
\end{proposition}
\begin{proof}
    For i.), Lemma \ref{lemma:one} gives $\text{LCD}\leq\text{FID}$ if and only if LCD is bi-independent, which by Proposition \ref{prop:defdiff} is equivalent to $\Delta_\text{LCD}=0$. Similarly, $\text{FID}\leq\text{FIID}$ if and only if FID is bi-independent and this is equivalent to $\Delta_\text{FID}=0$. Moreover, FIID may or may not be bi-independent. 
    
    For ii.), FDE networks are subclasses of ILC-1 networks and so by the preceding, $\Delta_\text{LCD}=\Delta_\text{FID}=0$. Furthermore, $\text{FID}=\text{FIID}$ if and only if the FIID is bi-independent and equivalently, $\Delta_\text{FIID}=0$.
\end{proof}

Recall that under ILC-1, the reaction networks are those in which the coarsest decomposition corresponds to the LCD, followed by the FID, the finest decomposition being the FIID.

\subsubsection{The subclass of FDE networks}

In some cases,the network's FID and FIID coincide, these is a special subclass of ILC-1 networks. We provide a sufficient condition for a network to attain FDE and that is, to be a deficiency zero network.

\begin{theorem}
\label{FDC:coincidence}
For any deficiency zero network, FID = FIID.
\end{theorem}

\begin{proof}
Suppose that the reaction network has deficiency zero. We will work with the restriction of the linear map $Y$ to ${{\rm{Im \ }} I_a}$, that is, the map given by $Y|_{{\rm{Im \ }} I_a} : {{\rm{Im \ }} I_a} \to  S$. By definition of the stoichiometric subspace $S$, this is a surjective map. Observe that ${\rm{ker \ }} Y|_{{\rm{Im \ }} I_a} = {\rm{ker \ }} Y \cap {{\rm{Im \ }} I_a}$ and so, $\delta = 0$ if and only if ${\rm{ker \ }} Y \cap {{\rm{Im \ }} I_a} = 0$. It follows that $Y|_{{\rm{Im \ }} I_a} : {{\rm{Im \ }} I_a} \to  S$ is an isomorphism.

The FID induces a sum ${\rm{Im \ }}I_a = {\rm{Im \ }}I_{a,1}+ {\rm{Im \ }}I_{a,2} + \cdots + {\rm{Im \ }}I_{a,k}$. By definition of $Y$, each $Y({\rm{Im \ }}I_{a,i})$ is contained in $S_i$. Since $Y|_{{\rm{Im \ }} I_{a,i}}$ is an isomorphism, each ${\rm{Im \ }}I_{a,i}$ is isomorphic to $S_i$. Therefore, $n_i -\ell_i = s_i$ for each $i=1,2,\ldots,k$. Summing up, we get $n - \ell$ on the left-hand side and $s$ on the right-hand side, showing that FID is bi-independent. 
By Theorem \ref{theorem:FDC:equiv}, the network is FDC.

To show that FID = FIID, we prove by contradiction. Suppose $\textup{FID}\ne\textup{FIID}$. Since the network is FDC, there is an incidence independent refinement of the FID. Without loss of generality, we can assume ${\rm{Im \ }}I_{a,i}$ is the direct sum of ${\rm{Im \ }}I'_{a,i}$ and ${\rm{Im \ }}I''_{a,i}$. Then, their images under $Y$ can be written as $S_i = S_i' + S_i''$. If this sum is
not direct, then there is a non-zero element that is the image under $Y|_{\textup{Im} \ I_a}$ of distinct elements in ${\rm{Im \ }}I'_{a,i}$ and ${\rm{Im \ }}I''_{a,i}$, a contradiction to the bijectivity of $Y|_{\textup{Im} \ I_a}$. On the other hand, if the sum is direct, it is a
contradiction to the FID.
\end{proof}

\begin{corollary}
    A network has zero deficiency if and only if its FID is a zero deficiency decomposition (ZDD).
    \label{ZDD}
\end{corollary}

\begin{proof}
    ``$\Leftarrow$'' For any independent decomposition, the sum of the subnetwork deficiencies is greater
or equal to the network deficiency. By assumption, these are all zero, implying the claim.
``$\Rightarrow$''  By Theorem \ref{FDC:coincidence}, we have FID=FIID. Hence, the network deficiency is the sum of the subnetwork deficiencies. Since the left-hand side is zero, it follows that each summand on the right-hand side must also be zero.
\end{proof}

\begin{remark}
It is easy to see that the right-hand side of Corollary \ref{ZDD} is equivalent to the statement that any
independent decomposition is ZDD, so that the left-hand side is just the special case of the trivial
decomposition. What is surprising is that this special case implies that the general statement. Theorem \ref{FDC:coincidence} shows that this is due to the FDE property.
\end{remark}

We also provide a characterization for FDE networks.

\begin{proposition}
The following statements are equivalent:
\begin{enumerate}
    \item[i.] A CRN is an FDE network, i.e., FID  = FIID.
    \item[ii.] FID and FIID are bi-independent.
    \item[iii.] Every independent and incidence-independent network is bi-independent.
    \item[iv.] The network deficiency is the sum of the subnetwork deficiencies for every independent and every incidence-independent network.
\end{enumerate}
\end{proposition}

\begin{proof}
We can divide each statement into a substatement about independent and incidence-independent decompositions, e.g., statement i) is equivalent to FID is a coarsening of FIID and vice versa. The first substatements are identical to those for FDC networks and hold, since each FDE network is FDC. The remaining sub-statements are derived by simply interchanging  independence and incidence-independence, FID and FIID in the arguments.    
\end{proof}

We know that deficiency zero networks produce FDE networks. We use matrix theory to provide network construction for deficiency zero networks. In particular, we show that a network with invertible $Y^{\top} Y$ has deficiency zero:

\begin{proposition}
For any reaction network, if $Y^{\top} Y$ is invertible, then $\delta=0$.
\end{proposition}

\begin{proof}
If $Y^{\top} Y$ is invertible, then ${\rm{ker}} (Y^{\top} Y) = 0$. Since ${\rm{ker \ }} Y$ is contained in ${\rm{ker}} (Y^{\top} Y)$, it follows that ${\rm{ker \ }} Y = 0$. Therefore, ${\ker \ } Y \cap { \ \rm{Im \ }} I_a = 0$ and equivalently,  $\delta =0$. 
\end{proof}

\begin{example}
    Consider the translated Anderies CRN with stoichiometric matrix given by
    \[
    \begin{blockarray}{ccccc}
        \matindex{$R_1$} & \matindex{$R_2$} & \matindex{$R_3$} & \matindex{$R_4$}\\
        \begin{block}{[cccc]c}
            -1 & 1 & 0 & 0 & \matindex{$A_1$}\\
            1 & -1 & -1 & 1 & \matindex{$A_2$}\\
            0 & 0 & 1 & -1 & \matindex{$A_3$}\\
        \end{block}
    \end{blockarray}
    \]
     which has rank 2. Note that $Y^\top Y=I_3$, which is invertible and the deficiency is $3-1-2=0$. 
\end{example}

A natural question to ask is if a reaction network's finest decomposition FID and FIID coincide, is $Y^\top Y$ invertible? The following example shows that the class of networks with invertible $Y^{\top} Y$ is a proper subclass of networks with deficiency equal to zero:

\begin{example}
      Consider the power law system (of type PL-RDK) representing a model of the Earth's carbon cycle with direct air capture (DAC) \cite{fortun3} has the following CRN:
\allowdisplaybreaks
\begin{multicols}{2}
 \noindent
 \begin{align*}
    &R_1: A_1 + 2A_2 \rightarrow 2A_1 + A_2\\
    &R_2: 2A_1 +A_2 \rightarrow A_1 + 2A_2 \\
    &R_3: A_2 \rightarrow A_3 \\
    &R_4: A_3 \rightarrow A_2 \\
    &R_5: A_4 \rightarrow A_{2} \\
    &R_6: A_{2} \rightarrow A_5 \\
    &R_7: A_5 \rightarrow A_4.
\end{align*}
\end{multicols}

      The CRN has more complexes than the number of species. Hence, the left generalized inverse of its molecularity matrix $Y$ does not exist by Proposition \ref{prop:tan}. It follows that $Y^\top Y$ is not invertible.

      However, note that the stoichiometric matrix of DAC CRN is given by
      \[
      \begin{blockarray}{cccccccc}
      \matindex{$R_1$} & \matindex{$R_2$} & \matindex{$R_3$} & \matindex{$R_4$} & \matindex{$R_5$} & \matindex{$R_6$} & \matindex{$R_7$}\\
      \begin{block}{[ccccccc]c}
          1 & -1 & 0 & 0 & 0 & 0 & 0 & \matindex{$A_1$}\\
          -1 & 1 & -1 & 1 & 1 & -1 & 0 & \matindex{$A_2$}\\
          0 & 0 & 1 & -1 & 0 & 0 & 0 & \matindex{$A_3$}\\
          0 & 0 & 0 & 0 & -1 & 0 & 1 & \matindex{$A_4$}\\
          0 & 0 & 0 & 0 & 0 & 1 & -1 & \matindex{$A_5$}\\
      \end{block}    
      \end{blockarray}
      \]
      which has rank 4. Hence, the DAC CRN has deficiency equal to $6-2-4=0$. Therefore, DAC CRN satisfies $\textup{FID}=\textup{FIID}$ by Theorem \ref{FDC:coincidence}.
\end{example}

Some other equivalent statements are presented in Appendix \ref{app:mattheo}. We already know that deficiency zero networks are FDE and so we ask: Is the converse true? The following example demonstrates that some positive-deficiency networks can also be FDE.

\begin{example}
    Consider the network

{
\begin{center}
\begin{tikzpicture}
        \tikzset{vertex/.style = {minimum size=4em}}
        \tikzset{edge/.style = {->,> = {Stealth[length=2mm, width=2mm]}}}
        \tikzset{edge1/.style = {bend left,->,> = {Stealth[length=2mm, width=2mm]}, line width=0.20mm}}
        \tikzset{edge2/.style = {bend right,->,> = {Stealth[length=2mm, width=2mm]}, line width=0.20mm}}
        \node[vertex] (C) at (0,0) {$2X_1+X_2$};
        \node[vertex] (Z) at (2,0) {$3X_1$};
        \node[vertex] (A) at (4,0) {$X_1+2X_2$};
        \node[vertex] (B) at (6,0) {$3X_2$};
        \draw[edge, thick] (C) to (Z)
        node[above,xshift=-9.5mm] {$R_1$};
        \draw[edge, thick] (Z) to (A)
        node[above,xshift=-11mm] {$R_2$};
        \draw[edge, thick] (A) to (B)
        node[above,xshift=-9.5mm] {$R_3$};
        \draw[edge1, thick] (B) to (C)
        node[below,xshift=21mm,yshift=-0.5mm] {$R_4$};
        
\end{tikzpicture}
\end{center}
}

    with incidence matrix given by
    \[
      \begin{blockarray}{ccccc}
      \matindex{$R_1$} & \matindex{$R_2$} & \matindex{$R_3$} & \matindex{$R_4$}\\
      \begin{block}{[cccc]c}
          -1 & 0 & 0 & 1 & \matindex{$2X_1+X_2$}\\
          1 & -1 & 0 & 0 & \matindex{$3X_1$}\\
          0 & 0 & 1 & -1 & \matindex{$3X_2$}\\
          0 & 1 & -1 & 0 & \matindex{$X_1+2X_2$}\\
      \end{block}    
      \end{blockarray}
      \]
      having the trivial FIID. On the other hand, the stoichiometric matrix is given by
      \[
      \begin{blockarray}{ccccc}
      \matindex{$R_1$} & \matindex{$R_2$} & \matindex{$R_3$} & \matindex{$R_4$}\\
      \begin{block}{[cccc]c}
          1 & -2 & -1 & 2 & \matindex{$X_1$}\\
          -1 & 2 & 1 & -2 & \matindex{$X_2$}\\
      \end{block}    
      \end{blockarray}
      \]
      with rank equal to 1. This network also has the trivial FID. Thus, the network is an FDE. However, the deficiency is given by $\delta=4-1-1=2>0$.
\end{example}

\subsubsection{Non-FDE ILC-1 networks}
We provide other examples of ILC-1 networks that are non-FDE.

\begin{example}\label{NonFDEILC1Example}
    Consider the following reaction network:
    \begin{align*}
    2X_1 \quad{\stackrel{R_2}\leftarrow}\quad X_1\quad{\stackrel{R_1}\rightarrow} \quad  0\quad{\stackrel{R_3}\leftarrow}\quad X_2\quad{\stackrel{R_4}\rightarrow}\quad2X_2.
    \end{align*}
    
    Observe that the LCD is the trivial decomposition. The stoichiometric matrix is given by 
    \[
      \begin{blockarray}{ccccc}
      \matindex{$R_1$} & \matindex{$R_2$} & \matindex{$R_3$} & \matindex{$R_4$}\\
      \begin{block}{[cccc]c}
          -1 & 1 & 0 & 0 & \matindex{$X_1$}\\
          0 & 0 & -1 & 1 & \matindex{$X_2$}\\
      \end{block}    
      \end{blockarray}
      \]
      having the FID: $\{\{R_1,R_2\},\{R_3,R_4\}\}$. The incidence matrix is given by
      \[
      \begin{blockarray}{ccccc}
      \matindex{$R_1$} & \matindex{$R_2$} & \matindex{$R_3$} & \matindex{$R_4$}\\
      \begin{block}{[cccc]c}
          -1 & -1 & 0 & 0 & \matindex{$X_1$}\\
          0 & 1 & 0 & 0 & \matindex{$2X_1$}\\
          0 & 0 & -1 & -1 & \matindex{$X_2$}\\
          0 & 0 & 0 & 1 & \matindex{$2X_2$}\\
      \end{block}    
      \end{blockarray}
      \]
      having the FIID: $\{\{R_1\},\{R_2\},\{R_3\},\{R_4\}\}$. Therefore, this network is an ILC-1 network that is not an FDE.
\end{example}

\subsection{Characterization and examples for ILC-2 networks}
We first characterize the ILC-2 class:

\begin{proposition}
    Let $\mathcal N$ be a network with deficiency difference $\Delta$. Then $\mathcal N$ is an ILC-2 network if and only if $r(\Delta)=(=,>,=)$.
\end{proposition}
\begin{proof}
    Suppose that $\mathcal N$ is an ILC-2 network. By Lemma \ref{lemma:one}, $\text{FIID}<\text{FID}$ implies that FIID is bi-independent, but not the FID. Hence the ranges of the deficiency difference of the network's FIID and FID are $=$ and $>$, respectively. Furthermore, FIID being bi-independent follows that the network is an ILC which implies that $\Delta_\text{LCD}=0$. Therefore, we have that $r(\Delta)=(=,>,=)$. 

    Conversely, suppose that $r(\Delta)=(=,>,=)$. It is already known that $\text{LCD}\leq\text{FIID}$. Since $\delta=\delta_\text{FIID}$, the FIID is forced to be independent by Lemma \ref{lemma:two} implying that $\text{FIID}<\text{FID}$.
\end{proof}

A subset of trivial FIID networks forms an ILC-2 subclass:

\begin{proposition}
    Any network with trivial FIID and non-trivial FID is an ILC-2 network.
\end{proposition}
\begin{proof}
    A trivial FIID implies that it is bi-independent and that $\text{LCD}=\text{FIID}$. Thus, the network must be an ILC network. Furthermore, a non-trivial FID implies that $\text{FIID}<\text{FID}$. Note also that the network has a single linkage class.
\end{proof}

The following ILC-2 example has a non-trivial FIID and more the one linkage class:

\begin{example}\label{ILC3Example}
Consider the following ILC network with two linkage classes:\\
\begin{tikzpicture}[node distance=1.5cm, every node/.style={minimum size=1cm}, ->, >=Stealth]
  \node (A) {$X_n$};
  \node (B) [right of=A, node distance=2cm] {$X_n + X_{n+1}$};
  \node (C) [right of=B, node distance=2cm] {$X_{n+1}$};
  \node (D) [below of=B] {$0$};
  \node (E) [below of=D] {$X_m + X_{m+1}$};
  \node (F) [left of=E, node distance=2cm] {$X_m$};
  \node (G) [right of=E, node distance=2.1cm] {$X_{m+1}$};

  \draw (A) -- node[midway, above] {$R_2$} (B);
  \draw (B) -- node[midway, above] {$R_3$} (C);
  \draw (C) -- node[midway, right] {$R_4$} (D);
  \draw (D) -- node[midway, left] {$R_1$} (A);
  \draw (D) -- node[midway, left] {$R_5$} (F);
  \draw (F) -- node[midway, below] {$R_6$} (E);
  \draw (E) -- node[midway, below] {$R_7$} (G);
  \draw (G) -- node[midway, right] {$R_8$} (D);

\begin{scope}[xshift=6cm,yshift=-1.5cm]
    \node (P1) {$X_p$};
    \node (P2) [right of=P1, node distance=1.5cm] {$2X_p$};

    \node at ($(P1)!0.5!(P2)$) {$\rightleftharpoons$};

    \node at ($(P1)!0.5!(P2) + (0,0.4)$) {$R_9$};
    \node at ($(P1)!0.5!(P2) + (0,-0.4)$) {$R_{10}$};
  \end{scope}
\end{tikzpicture}

where $m,n,p\in\mathbb N$ are distinct. The decompositions of the network are as follows:
\begin{enumerate}
    \item $\text{LCD}:\{\{R_1,R_2,R_3,R_4,R_5,R_6,R_7,R_8\},\{R_9,R_{10}\}\}$
    \item $\textup{FID}:\{\{R_1,R_3\},\{R_2,R_4\},\{R_5,R_7\},\{R_6,R_8\},\{R_9,R_{10}\}\}$
    \item $\text{FIID}: \{\{R_1,R_2,R_3,R_4\},\{R_5,R_6,R_7,R_8\},\{R_9,R_{10}\}\}$.
\end{enumerate}
Hence, the network exhibits an ILC-2 network.
\end{example}

\subsection{Characterization and examples of ILC-3 networks}
This FDA class is characterized as follows:
\begin{proposition}
    Let $\mathcal N$ be a network with deficiency difference $\Delta$. Then $N$ is an ILC-3 network if and only if $r(\Delta)=(=,>,>)$.
\end{proposition}
\begin{proof}
    Suppose that $\mathcal N$ is an ILC-3 network. The condition $\text{LCD}<\text{FID}$ implies that $\mathcal N$ is an independent linkage class network, and hence the range of $\Delta_\text{LCD}$ is $=$. Since both FID and FIID are not ordered, both must not be bi-independent implying that the ranges of $\Delta_\text{FID}$ and $\Delta_\text{FIID}$ are $<$.

    Conversely, suppose that $r(\Delta)=(=,>,>)$. The assumption that $\Delta_\text{LCD}=0$ implies that $\mathcal N$ is ILC and so, $\text{LCD}<\text{FID}$. Since $\Delta_\text{FID}>0$ and $\Delta_\text{FIID}>0$, both FID and FIID cannot be bi-independent and hence, cannot be ordered.
\end{proof}

\begin{example} \label{ILC2Ex}
Consider the following reaction network:
\begin{center}
        \begin{tikzpicture}[node distance=1.5cm, every node/.style={minimum size=1cm}, ->, >=Stealth]
  \node (A) {$X_n$};
  \node (B) [right of=A, node distance=2cm] {$X_n + X_{n+1}$};
  \node (C) [right of=B, node distance=2cm] {$X_{n+1}$};
  \node (D) [below of=B] {$0$};
  \node (E) [below of=D] {$X_m + X_{m+1}$};
  \node (F) [left of=E, node distance=2cm] {$X_m$};
  \node (G) [right of=E, node distance=2.1cm] {$X_{m+1}$};

  \draw (A) -- node[midway, above] {$R_2$} (B);
  \draw (B) -- node[midway, above] {$R_3$} (C);
  \draw (C) -- node[midway, right] {$R_4$} (D);
  \draw (D) -- node[midway, left] {$R_1$} (A);
  \draw (D) -- node[midway, left] {$R_5$} (F);
  \draw (F) -- node[midway, below] {$R_6$} (E);
  \draw (E) -- node[midway, below] {$R_7$} (G);
  \draw (G) -- node[midway, right] {$R_8$} (D);

\begin{scope}[xshift=6cm,yshift=-1.5cm]
    \node (P1) {$X_p$};
    \node (P2) [right of=P1, node distance=1.5cm] {$2X_p$};

    \node at ($(P1)!0.5!(P2)$) {$\rightleftharpoons$};

    \node at ($(P1)!0.5!(P2) + (0,0.4)$) {$R_9$};
    \node at ($(P1)!0.5!(P2) + (0,-0.4)$) {$R_{10}$};
  \end{scope}

\end{tikzpicture}
\begin{tikzpicture}[node distance=1.5cm, every node/.style={minimum size=1cm}, ->, >=Stealth]
  \node (A) {$2X_q$};
  \node (B) [right of=A, node distance= 1.6cm] {$X_q$};
  \node (C) [right of=B, node distance= 1.6cm] {$0$};
  \node (D) [right of=C, node distance= 1.6cm] {$X_r$};
  \node (E) [right of=D, node distance= 1.6cm] {$2X_r$};

  \draw (B) -- node[midway, above] {$R_{11}$} (A);
  \draw (B) -- node[midway, above] {$R_{12}$} (C);
  \draw (D) -- node[midway, above] {$R_{13}$} (C);
  \draw (D) -- node[midway, above] {$R_{14}$} (E);
  \end{tikzpicture}
  \end{center}
  where $n,m,p,q,r\in\mathbb N$ are distinct. Observe that the network from Example \ref{NonFDEILC1Example} is a subnetwork of this network. The addition of reactions $R_{1}$ to $R_{10}$ enabled the reaction network to become an ILC-3 network.
\end{example}

A notable observation is that the network from the preceding example has Example \ref{ILC3Example} as its subnetwork, which suggests that an ILC-3 network can be derived from ILC-1 and ILC-2 networks.

The following proposition allows the construction of ILC-3 subclasses from those of ILC-1 and ILC-2:
\begin{proposition}\label{ILC-1+ILC-2}
    Let $(\mathcal{S}_1, \mathcal{C}_1, \mathcal{R}_1)$ and $(\mathcal{S}_2, \mathcal{C}_2, \mathcal{R}_2)$ be ILC-1 and ILC-2 networks, respectively. If their union is an ILC network, then $(\mathcal{S}_1\cup\mathcal{S}_2, \mathcal{C}_1\cup\mathcal{C}_2, \mathcal{R}_1\cup\mathcal{R}_2)$ is an ILC-3 network.
\end{proposition}

Proposition \ref{ILC-1+ILC-2} tells us that we can construct ILC-3 networks from ILC-1 and ILC-2. It is worthwhile investigating if it is possible to decompose ILC-3 networks into ILC-1 and ILC-2 subnetworks.

\section{Kinetic properties of ILC systems}
\label{kinetic properties:ILC}

\subsection{Application to the existence of positive equilibria}

One advantage of networks with the ILC property is that the existence of a set of positive equilibria for the entire network can be examined through its linkage classes. This is one of the motivations for studying network decompositions: by breaking the network into smaller and more manageable pieces, we can more easily determine the network properties, such as the existence of positive equilibria. Theorem \ref{result:boros} by B. Boros \cite{boros:thesis} provides a characterization of the existence of the set of positive equilibria of a CRN under mass-action kinetics based on its linkage classes. This was extended to a large class of power-law kinetic systems in Theorem \ref{extension:talabis} by Hernandez and Mendoza \cite{Hernandez:polyPLP}.

	\begin{theorem}
		Let $(\mathcal{N},K)$ be a mass-action system that satisfies $\delta=\delta_1+\cdots+\delta_\ell$. Then
		$$E_+ \ne \varnothing {\text{ if and only if }} E_+^i \ne \varnothing$$
		for each $i=1,\ldots,\ell.$
		\label{result:boros}
	\end{theorem}

    	\begin{theorem}
		Let $(\mathcal{N},K)$ be a PL-RDK system that satisfies $\delta=\delta_1+\cdots+\delta_\ell$ and ${\widehat{T}}={\widehat{T}}^1 \oplus \cdots \oplus {\widehat{T}}^\ell$ (which we call the $\widehat{T}$-independence). Then
		$$E_+ \ne \varnothing {\text{ if and only if }} E_+^i \ne \varnothing$$
		for each $i=1,\ldots,\ell.$
		\label{extension:talabis}
	\end{theorem}

In ILC networks, the LCD is a coarsening of the FID. Hence, one may check the existence of positive equilibria at the level of the linkage classes of the network. 

\begin{example}\label{example:ILCproperty}
    Consider the following network with the ILC property
    \begin{align*}
         A \underset{R_2}{\stackrel{R_1}\rightleftharpoons} A+B \quad \quad B \underset{R_4}{\stackrel{R_3}\rightleftharpoons} B+C. 
    \end{align*}

Assuming mass action kinetics, the first linkage class has associated positive equilibria of the form $a=\tau_1>0$ and $b=\dfrac{k_1}{k_2}$ while the second has the form $a=\tau_2>0$ and $b=\dfrac{k_3}{k_4}$. Note that the positive equilibria of the subnetworks exist for all rate constants. It follows the same for the entire network.

\end{example}

\subsection{Application to equilibria parametrization}

Theorem \ref{feinberg decomp} forms the foundation of the framework proposed by Hernandez et al. \cite{HernandezetalPCOMP2023} for computing the parametrization of positive equilibria of an entire network based on its independent subnetworks. Specifically, the theorem states that the set of positive equilibria of the full network is the intersection of the sets of positive equilibria of its independent subnetworks. Since we are focusing on ILC systems, this allows one to compute the set of positive equilibria for the entire network by analyzing the equilibria of the subsystems associated with individual linkage classes.

\begin{example}
    Consider the carbon cycle model with the direct ocean capture technology \cite{alaminDOC} as follows:
\[\begin{aligned}
    A_1 + 2A_2 \xrightleftharpoons{\qquad} A_2 + 2A_1\qquad \begin{tikzpicture}[node distance = {20mm}, baseline=(current  bounding  box.center)]
			\node (A17) {$A_{17}$};
            \node (A3) [above of=A17]{$A_3$};
            \node (A2) [left of=A3]{$A_2$};
            \node (A4) [left of=A17]{$A_4$};
            \draw[->] (A3) -- (A17);
            \draw[->] (A17) -- (A4);
            \draw[->] (A4) -- (A2);
            \draw [-left to] ($(A2.east) + (0pt, 2pt)$) -- ($(A3.west) + (0pt, 2pt)$);
    \draw [-left to] ($(A3.west) + (0pt, -2pt)$) -- ($(A2.east) + (0pt, -2pt)$);
        \end{tikzpicture},
\end{aligned}\]
 which describes the carbon cycle interactions through the transfer of carbon between the carbon pools: land ($A_1$), atmosphere ($A_2$), ocean ($A_3$), carbon stock ($A_4$) and direct air capture ($A_{17}$). This network has the ILC property, so the set of positive equilibria can be computed independently for each linkage class, and the intersection of these sets yields the set of positive equilibria for the entire network \cite{alaminDOC}.
\end{example}

\subsection{The single linkage class case}

Clearly, if the underlying network of a kinetic system has a single linkage class, then Theorem \ref{extension:talabis} is not
very useful. However, determining its FDA class can be quite beneficial. If the FDA class is ILC-2 or ILC-3, the FID is guaranteed to be a proper refinement, and hence Theorem \ref{feinberg decomp} can be fruitfully used for equilibria computation. This possibility is particularly important for large networks.

In the general case, we have the following Corollary for Theorem \ref{extension:talabis}:

\begin{corollary}
Let $(\mathcal{N}, \mathcal{K})$ be a $\widehat{T}$-independent kinetic system (as in Theorem \ref{extension:talabis}). If $\mathcal{L}_i$ is a linkage class of
FDA class ILC-2 or ILC-3, then its FID can be used to facilitate the analysis of equilibria for the entire
network.
\end{corollary}

\section{FDA classes of DLC networks}
\label{section:FDA:DLC}

In this section, we first show that the deficiency difference only partially characterizes the FDA classes of DLC networks. However, this is easily remedied by additionally considering the common complexes cardinality or $\mathscr{CC}$-cardinality of DLC-1 networks. Recall from Section \ref{LCD:C:decomposition} that we define the new concept of the set of common complexes as $\mathscr{CC}=\mathscr{CC}_{FID}$.  The combined use of deficiency difference $\Delta$ and $\mathscr{CC}$-cardinality results in a convenient procedure for computing the FDA class of any network. We provide numerous examples for the three DLC classes.


\subsection{Partial characterization of the FDA classes of DLC networks}
\label{partial:characterization:DLC}

The following proposition shows that the deficiency difference characterizes DLC-3 networks, but not DLC-1 and DLC-2 networks:

\begin{proposition}
\label{prop:delta:dlc3}
    For any $\mathcal N$ with deficiency difference $\Delta$,
    \begin{enumerate}
        \item[i.] $\mathcal N$ is a DLC-3 network if and only if $r(\Delta)=(>,>,>)$.
        \item[ii.] if $\mathcal N$ is a DLC-1 network, then $r(\Delta)=(>,=,>)$.
        \item[iii.] if $\mathcal N$ is a DLC-2 network, then $r(\Delta)=(>,=,>)$.
    \end{enumerate}
\end{proposition}

\begin{proof}
    For (ii) and (iii), in both DLC-1 and DLC-2 architectures, the network is by definition a Dependent Linkage Class (DLC) system, which implies that the LCD is not independent. Consequently, $\Delta_{LCD} > 0$ and the Finest Incidence-Independent Decomposition (FIID) cannot be independent, otherwise its coarsening, the LCD, would also be independent by Lemma \ref{lemma:one}, forcing $\Delta_{FIID} > 0$. However, in both DLC-1 and DLC-2 classes, the Finest Independent Decomposition (FID) is a coarsening of the FIID, which according to Theorem \ref{theorem:FDC:equiv} ensures that the FID is bi-independent. Thus, $\Delta_{FID}=0$. For (ii), for DLC-3, the DLC property similarly ensures $\Delta_{LCD} > 0$ and $\Delta_{FIID} > 0$, but because the FID and FIID are not ordered, the FID cannot be incidence-independent and thus lacks bi-independence, resulting in $\Delta_{FID} > 0$ and the range $(>,>,>)$.
\end{proof}

This means that we still need to identify a network property differentiating DLC-1 and DLC-2 networks.

\subsection{$\mathscr{C}\mathscr{C}$-cardinality condition for DLC-1 networks
}
\label{section:CC:based}

In this section, we focus on DLC networks and show that their FDA classes can be characterized on the basis of the network's common complexes cardinality $|\mathscr{C}\mathscr{C}|$. Surprisingly to a greater extent, the network invariant plays a decisive role for DLC networks than that of deficiency for ILC networks. The computation of a DLC network's $\mathscr{C}\mathscr{C}$ enables a systematic determination of its FDA class.

\subsubsection{Zero $\mathscr{C}\mathscr{C}$-cardinality networks}

The next proposition shows that the DLC-1 networks are precisely the $\mathscr{C}\mathscr{C}$ cardinality zero networks:

\begin{proposition}
\label{prop:DLC1:CC0}
    Let $\mathcal{N}$ be a DLC network. Then $\mathcal{N}$ is a DLC-1 network if and only if $|\mathscr{CC}| = 0$.
\end{proposition}

\begin{proof}
$(\Leftarrow)$ We assume $|\mathscr{CC}|= 0$ for the Finest Independent Decomposition (FID). By definition, this makes the FID a $\mathscr{C}$-decomposition. From Theorem \ref{structure:theorem:farinas}, the LCD is a refinement of any $\mathscr{C}$-decomposition. Thus, the FID is a coarsening of the LCD. Since the network is DLC, LCD $<$ FID (i.e., the FID is of course independent while the LCD is dependent), which forces the proper coarsening relationship LCD $<$ FID. However, the LCD is always a coarsening of the FIID, which leads to the hierarchy FID $>$ LCD $>$ FIID.
    $(\Rightarrow)$ We assume that the network is DLC-1. By definition, the decomposition hierarchy is FID $>$ LCD $>$ FIID, which means that the LCD is a coarsening of the FID. By the Structure Theorem for $\mathscr{C}$-decompositions, a decomposition is a coarsening of the LCD if and only if it is a $\mathscr{C}$-decomposition. Since the FID is a refinement of the LCD, and the LCD itself has a common complex cardinality of zero, the FID must also have no common complexes. Thus, $|\mathscr{CC}|=0$.
\end{proof}

\subsubsection{A procedure for determining the FDA class of a reaction network}
\label{procedure:FDA}

We now describe a simple procedure to determine the FDA class of any reaction network:
\begin{enumerate}
    \item Compute the LCD, FID and FIID of the network. This can be done for the LCD with CRNToolbox and for the FID and FIID with DECENT.
    \item Compute the deficiency difference $\Delta$. The deficiency information is also available from the tools in Step 1.
    \item If $r(\Delta) \ne (>, =, >)$, then the network's FDA class is exactly one of the ILC classes or DLC-3.
    \item Otherwise, compute the $\mathscr{CC}$-cardinality. Since the network is DLC, if this is zero, the class is DLC-1, otherwise, DLC-2.
\end{enumerate}

A summary of important computed values are given in Tables \ref{deficiency:difference}, \ref{C:decomposition}, and \ref{deficiency:difference:cc:cardinalities} for six selected biochemical networks.

\begin{table}[h]
    \centering
    \renewcommand{\arraystretch}{1.5} 
    \begin{tabular}{|p{4.5cm}|c|c|c|c|}
    \hline
    \textbf{Network} & $\delta$ & $\delta_{LCD}$ & $\delta_{FID}$ & $\delta_{FIID}$ \\
    \hline
    INSMS & 7 & 0 & 7 & 0 \\
    \hline
    INRES & 19 & 0 & 19 & 0 \\
    \hline
    Schmitz Wnt & 2 & 0 & 2 & 0 \\
    \hline
    Feinberg-Lee Wnt & 2 & 0 & 2 & 0 \\
    \hline
    MacLean Wnt & 4 & 0 & 4 & 0 \\
    \hline
    A power law model of a complex coordination of multi-scale cellular responses to environmental stress & 1 & 0 & 1 & 0 \\
    \hline
    \end{tabular}
    \caption{The deficiency of the network $\delta$ and the sum of the deficiencies of the subnetworks under the LCD, FID and FIID of six selected biochemical networks.}
    \label{deficiency:difference}
\end{table}

\begin{table}[h]
    \centering
    \renewcommand{\arraystretch}{1.5} 
    \begin{tabular}{|p{4.5cm}|c|c|}
    \hline
    \textbf{Network} & $\mathscr{C}\mathscr{C}$ & $|\mathscr{C}\mathscr{C}|$ \\
\hline
 INSMS  & $\{0, X_6, X_{13}, X_{20}\}$ &  4\\
 \hline
 INRES  & $\{\}$ & 0\\
 \hline
 Schmitz Wnt  & $\{A_1\}$ & 1\\
 \hline
 Feinberg-Lee Wnt  & $\{0,A_2,A_{23}\}$ & 3\\
 \hline
 MacLean Wnt  & $\{A_{13}\}$ & 1\\
 \hline
    A power law model of a complex coordination of multi-scale cellular responses to environmental stress & $\{X_2+X_7, X_3, X_4, X_8, X_9\}$ & 5\\
    \hline
    \end{tabular}
    \caption{The set of common complexes $\mathscr{CC}$ and the $\mathscr{CC}$-cardinalities of six selected biochemical networks.}
    \label{C:decomposition}
\end{table}

\begin{table}[h]
    \centering
    \renewcommand{\arraystretch}{1.5} 
    \begin{tabular}{|p{4.5cm}|c|c|c|c|}
    \hline
    \textbf{Network} & $\Delta$ & $r(\Delta)$ & $|\mathscr{C}\mathscr{C}|$ & FDA class \\
    \hline
    INSMS & (7,0,7) & $(>,=,>)$ & 4 & DLC-2 \\
    \hline
    INRES & (19,0,19) & $(>,=,>)$ & 0 & DLC-1 \\
    \hline
    Schmitz Wnt & (2,0,2) & $(>,=,>)$ & 1 & DLC-2 \\
    \hline
    Feinberg-Lee Wnt & (2,0,2) & $(>,=,>)$ & 3 & DLC-2 \\
    \hline
    MacLean Wnt & (4,0,4) & $(>,=,>)$ & 1 & DLC-2 \\
    \hline
    A power law model of a complex coordination of multi-scale cellular responses to environmental stress & (1,0,1) & $(>,=,>)$ & 5 & DLC-2 \\
    \hline
    \end{tabular}
    \caption{A summary of the deficiency differences $\Delta$, range of values $r(\Delta)$, $\mathscr{CC}$-cardinalities, and FDA classes of six selected biochemical networks.}
\label{deficiency:difference:cc:cardinalities}
\end{table}

\subsubsection{Examples of DLC-1 systems}

DLC-1 networks include all networks with $|\mathscr{C}\mathscr{C}|=0$, i.e., subnetworks have no common complexes. It contains the subclass of all networks with only the trivial decomposition as independent and there are at least two linkage classes. Furthermore, rank 1 networks with at least two linkage classes are also examples.

\begin{example}
Consider the following rank 1 network with two linkage classes:
$$0 \to 3A \quad 2A\to A.$$
Since the FID contains only the network itself with the two reactions, it follows that the network is of DLC-1.
\end{example}

\begin{example}\label{ex:SFsub}
    Consider the subnetwork of the Shinar-Feinberg network of calcium signaling in olfactory cilia \cite{feinberg:book} given by
\begin{align*}
    B\ &{\stackrel{R_3}\rightarrow}\ D+B\\
    D\ &{\stackrel{R_9}\rightarrow}\ 0.
    \end{align*}
Note that $\text{LCD}:\{\{R_3\},\{R_9\}\}$, $\textup{FIID}:\{\{R_3\},\{R_9\}\}$, and $\textup{FID}:\{\{R_3,R_9\}\}$. Thus, the network is of DLC-1.

Let us verify this result using the procedure in Section \ref{procedure:FDA}.
The network deficiency is $\delta = 1$ and the subnetwork deficiency sums are $\delta_{LCD} = 0$, $\delta_{FID} = 1$, and $\delta_{FIID} = 0$. Thus, the resulting deficiency difference vector is $\Delta = (1, 0, 1)$. Thus, $r(\Delta) = (>, =, >)$, which indicates that the network is not one of the ILC classes or DLC-3. At this point, the network must be DLC-1 or DLC-2. The subnetworks under the FID share no common complexes. Hence, $|\mathscr{CC}|=0$. Indeed, the network is of DLC-1.


\end{example}

\begin{example}
\label{example:INRES}
     Consider the model of insulin signaling in type 2 diabetes cell (INRES) \cite{Nyman} with 44 reactions provided in Appendix \ref{appendix:inres}. From Table \ref{C:decomposition}, the CRN has $|\mathscr{CC}|=0$ and it is a DLC-1 network.
\end{example}

A prominent subclass of DLC-1 networks consists of those whose trivial decomposition is its only independent decomposition, i.e., equals its FID, and has multiple linkage classes. Such networks of higher rank are difficult to analyze—an example for this is the terrestrial CDR (carbon dioxide removal) model of Heck et al. studied in \cite{Hernandez:multi}.

\subsection{Examples of DLC-2 and DLC-3 networks}

The result in Section \ref{partial:characterization:DLC} shows that DLC-2 and DLC-3 networks can be differentiated through their deficiency difference. In this section, we document a further difference: DLC-2 contains both low and higher $\mathscr{CC}$-cardinality networks while DLC-3 consists only of the latter. In analogy to low and higher deficiency, ``low'' means 0 or 1 while ``higher'' means 2 or greater.


\subsubsection{$\mathscr{CC}$-cardinality one networks}

The following proposition identifies a subclass of DLC-2:
\begin{proposition}
    \label{prop:DLC2:CC1}
    Any $\mathscr{CC}$-cardinality one network is contained in DLC-2.
\end{proposition}

\begin{proof}
By definition, the network has $|\mathscr{C}\mathscr{C}|=1$ if and only if the FID is a $\mathscr{C}*$-decomposition. By Proposition \ref{cdecomposition:incidenceindependent}, this implies that the FID is incidence independent, i.e., a coarsening of the FIID. Again this coarsening is proper, i.e., the FIID is not independent because it has a dependent coarsening, namely the LCD, since the network is DLC. Hence, the network has DLC-2 FDA.
\end{proof}

We know that a subnetwork of the Shinar-Feinberg network in Example \ref{ex:SFsub} can exhibit DLC-1 behavior. Interestingly, the next example shows that considering the entire network reveals a different class of DLC structure.

\begin{example}
    Consider the Shinar-Feinberg network \cite{feinberg:book} given by
    \begin{align*}
        A\ &{\underset{R_2}{\stackrel{R_1}\rightleftharpoons}}\ B\ {\stackrel{R_3}\rightarrow
        D+B}\\
        C+4D\ &{\underset{R_5}{\stackrel{R_4}\rightleftharpoons}} \ E\\
        B+E\ &\stackrel{R_6}\rightarrow \ F\ {\underset{R_8}{\stackrel{R_7}\rightleftharpoons}}\ A+E\\
        D\ & \stackrel{R_9}\rightarrow \ 0.
    \end{align*}
The decompositions of the network are as follows:
\begin{enumerate}
    \item $\text{LCD}:\{\{R_1,R_2,R_3\},\{R_4,R_5\},\{R_6,R_7,R_8\},\{R_9\}\}$
    \item $\textup{FID}:\{R_1,R_2,R_6,R_7,R_8\},\{R_3,R_9\},\{R_4,R_5\}\}$
    \item $\text{FIID}: \{\{R_1,R_2\},\{R_3\},\{R_4,R_5\},\{R_6\},\{R_7,R_8\},\{R_9\}\}$,
\end{enumerate}
which follows a class DLC-2 network.
\end{example}

\begin{example}
    Consider the network of Hernandez et al. \cite{HernandezetalPCOMP2023} as follows
$$2Y \stackrel{R_4}\rightarrow X\ {\underset{R_2}{\stackrel{R_1}\rightleftharpoons}}Y \ \ \ \ \ 2X \stackrel{R_3}\rightarrow  X+Y.$$
The decompositions of the network are as follows:
\begin{enumerate}
    \item $\text{LCD}:\{\{R_1,R_2,R_4\},\{R_3\}\}$
    \item $\textup{FID}:\{\{R_1,R_2\},\{R_3,R_4\}\}$
    \item $\text{FIID}: \{\{R_1,R_2\},\{R_3\},\{R_4\}\}$.
\end{enumerate}
Hence, the network is class DLC-2 network.
\label{network:Hernandez}
\end{example}

\begin{example}
The two multistationary Wnt signaling models Schmitz and MacLean \cite{HernandezetalECREWS,Maclean,schm2002} have $|\mathscr{C}\mathscr{C}|=1$ as seen in Table \ref{C:decomposition} and hence are also in DLC-2.
\end{example}

\subsubsection{Higher $\mathscr{CC}$-cardinality networks in DLC-2}

\begin{example}
        Consider the CRN for the 2-site  phosphorylation/dephosphorylation network \cite{CODH2018} given by
    \begin{align*}
        &S_0 + K \underset{R_2}{\stackrel{R_1}\rightleftharpoons} S_0K {\stackrel{R_3}\rightarrow} S_{1} + K \underset{R_5}{\stackrel{R_4}\rightleftharpoons} S_1K {\stackrel{R_6}\rightarrow} S_{2} + K\\
        &S_{2} + F \underset{R_{8}}{\stackrel{R_{7}}\rightleftharpoons} S_{2}F \ {\stackrel{R_{9}}\rightarrow} S_{1} + F \underset{R_{11}}{\stackrel{R_{10}}\rightleftharpoons} S_{1}F \ {\stackrel{R_{12}}\rightarrow} S_{0} + F
    \end{align*}
The decompositions of the network are as follows:
\begin{enumerate}
    \item $\text{LCD}:\{\{R_1,R_2,R_3,R_4,R_5,R_6\},\{R_7,R_8,R_9,R_{10},R_{11},R_{12}\}\}$
    \item $\textup{FID}:\{\{R_1,R_2,R_3,R_{10},R_{11},R_{12}\},\{R_{4},R_{5},R_{6},R_7,R_8,R_9\}\}$
    \item $\text{FIID}: \{\{R_1,R_2\},\{R_3\},\{R_4,R_5\},\{R_6\},\{R_7,R_8\},\{R_9\},\{R_{10},R_{11}\},\{R_{12}\}\}$.
\end{enumerate}
Hence, the network is a class DLC-2 network.
\label{network:PDnetwork}
\end{example}

\begin{example}
    Consider the Feinberg-Lee model \cite{feinberg:book} with reactions given in Appendix \ref{appendix:Wnt}. From Table \ref{C:decomposition}, $|\mathscr{C}\mathscr{C}|=3 \ge 1$. The deficiency of the network is two. Furthermore, the deficiencies of the subnetworks under the FIDs sum up to network's deficiency. Hence, $\Delta_{FID}=0$ and it belongs to DLC-2. Indeed, we can verify that the FID is a coarsening of the FIID, i.e., the FID contains a subnetwork $\{R_1, R_2, R_3, R_4, R_5, R_{12}, R_{13}\}$ while these two from separate linkage classes $\{R_1, R_2, R_5, R_{14}, R_{15}\}$ and $\{R_3, R_4, R_{12}, R_{13}\}$.
\end{example}

\begin{example}
    Consider the power law model of a complex coordination of multi-scale cellular responses \cite{AJLM2017,fonseca} with reactions provided in Appendix \ref{appendix:complex:coordination}. From Table \ref{C:decomposition}, we have $|\mathscr{C}\mathscr{C}|=5 \ge 1$. The deficiency of the network is one. Furthermore, the deficiencies of the subnetworks under the FIDs of these networks sum up to network's deficiency. Hence, $\Delta_{FID}=0$ and it belongs to DLC-2. Indeed, the FID has one subnetwork $\{R_5, R_6, R_7, R_8\}$ while these three form separate linkage classes $\{R_3, R_4, R_5, R_6\}$, $\{R_7\}$ and $\{R_8\}$.
\end{example}

\begin{example}
    Consider the model of insulin signaling in healthy cells (INSMS) by Sedaghat et al. \cite{Sedaghat}. The reactions are provided in Appendix \ref{appendix:insms}. The CRN was initially explored by Lubenia et al. \cite{LML2023}, who investigated fundamental concepts in Chemical Reaction Network Theory, specifically absolute concentration robustness and concordance. Their analysis yielded the system's FID. From a coarsening of this FID, they identified three subnetworks that are both functionally and structurally significant. In this network, the deficiencies of the subnetworks under the FIDs sum up to network's deficiency so $\Delta_{FID}=0$.
    Furthermore, based on Table \ref{C:decomposition}, we have $|\mathscr{CC}|=4 \ge 1$. Thus, the network belongs to DLC-2. 
\end{example}

\subsubsection{DLC-3 networks}

The following Proposition collects two basic properties of DLC-3 networks.
\begin{proposition}
Let $\mathcal{N}$ be a DLC-3 network. Then
\begin{enumerate}
    \item[i.] its deficiency is greater than zero, i.e., $\delta >0$.
    \item[ii.] its $\mathscr{CC}$-cardinality is greater than one, i.e., $|\mathscr{CC}|>1$.
\end{enumerate}
\end{proposition}

\begin{proof}
(i) By Proposition \ref{prop:delta:dlc3}, a network is DLC-3 if and only if its deficiency difference range is $r(\Delta) = (>, >, >)$. This implies $\Delta_{FID} = |\delta - \delta_{FID}| > 0$. By Theorem \ref{theorem:FDC:equiv}, the FID is not bi-independent. From Theorem \ref{FDC:coincidence}, all deficiency zero networks are FDE (i.e., FID = FIID), which are bi-independent. Since the network is not bi-independent, then $\delta > 0$. (ii) Propositions \ref{prop:DLC1:CC0} and \ref{prop:DLC2:CC1} state that all low $\mathscr{CC}$-cardinality networks are contained in DLC-1 and DLC-2, leaving only higher $\mathscr{CC}$-cardinality networks for DLC-3.
\end{proof}

Note that while the first necessary condition is shared with DLC-1 and DLC-2. On the other hand, the second one is unique to DLC-3.

\begin{example}
    The Mycobacterium tuberculosis NRP S-system \cite{fari2020,Magombedze} network with 80 reactions in Appendix \ref{appendix:tb:nrp} is of DLC-3.
\end{example}

\subsection{FDA impact on DLC network analysis}

For checking the existence of equilibria and their parametrization in DLC systems, the linkage classes are dependent and thus cannot serve as a suitable basis. Instead, one must use independent subnetworks of the network, which are of course not associated to linkage classes, in accordance with Theorem \ref{feinberg decomp}. This is especially important for networks of DLC-2 or DLC-3, where the Linkage Class Decomposition (LCD) and the FIID are not aligned. For instance, consider the network below, which is a DLC-2.

\begin{example}
    Reconsider the DLC-2 network in Example \ref{network:Hernandez}. It has the FID with two independent subnetworks $\mathcal{N}_1$ with reactions $R_1$ and $R_2$, and $\mathcal{N}_2$ with reactions $R_3$ and $R_4$. With the assumption of mass action kinetics, Hernandez et al. \cite{HernandezetalPCOMP2023} computed the positive equilibria to be $\left(\tau_1, \dfrac{k_1}{k_2}\tau_1 \right)$ for $\mathcal{N}_1$ and $\left(\tau_2, \sqrt{\dfrac{k_3}{k_4}}\tau_2 \right)$ for $\mathcal{N}_2$ for positive free parameters $\tau_1$ and $\tau_2$. The computations are performed on the subnetworks of the FID rather than on its linkage classes. Furthermore, both subnetworks admit positive steady states for arbitrary values of the rate constants. However, the full network admits a positive steady state only when the condition $\dfrac{k_1}{k_2}=\sqrt{\dfrac{k_3}{k_4}}$ holds.
\end{example}

\begin{example}
    Reconsider the insulin resistance signaling (INRES) network \cite{Nyman} with 44 reactions, which is of DLC-1, in Example \ref{example:INRES}. Hernandez and Lubenia \cite{Hernandezet:nonmassaction} computed positive equilibria parametrization by individually getting the positive equilibria parametrization of the FID.
\end{example}

\section{Summary, Conclusion, and Recommendations}
\label{summary:conclusion:recommendation}

The Finest Decompositions' Architecture (FDA) framework provides a comprehensive classification system that addresses the misalignment between a reaction network's visual connectivity and the algebraic properties governing its long-term dynamics. We categorize all networks into six distinct architectures based on the coarsening relationships between the Linkage Class Decomposition (LCD), the Finest Independent Decomposition (FID), and the Finest Incidence-Independent Decomposition (FIID).

Table \ref{deficiency:additivity} establishes a foundational distinction, showing how deficiency additivity, which is a key property for simplified analysis, behaves differently across the architectures. The results in the table show that additivity over the LCD is a hallmark of all ILC networks. Furthermore, the additivity over the FID is characteristic of FDC networks (ILC-1, DLC-1, and DLC-2).

\begin{table}[h]
\centering
\caption{{The FDC Properties and Deficiency Additivity of FDA Classes}}
\label{tab:fdc_deficiency_properties_revised}
\begin{tabular}{@{}llll@{}}
\toprule
\textbf{FDA Class} & \textbf{FDC Property} & \textbf{Deficiency Additivity} & \textbf{Deficiency Additivity} \\
& & \textbf{over LCD} & \textbf{over FID} \\ \midrule
\textbf{ILC-1} & FDC & Additive & Additive \\
\textbf{ILC-2} & Non-FDC & Additive & Not Additive \\
\textbf{ILC-3} & Non-FDC & Additive & Not Additive \\ \midrule
\textbf{DLC-1} & FDC & Not Additive & Additive \\
\textbf{DLC-2} & FDC & Not Additive & Additive \\
\textbf{DLC-3} & Non-FDC & Not Additive & Not Additive \\ \bottomrule
\end{tabular}
\label{deficiency:additivity}
\end{table}

Aside from the primary classification, several foundational results were obtained including the FDC characterization. A network is identified as an FDC (Finest Decomposition Coarsening) network if and only if the FID is bi-independent, which ensures that the deficiency of the whole network equals the sum of its subnetwork deficiencies under any independent decomposition. Furthermore, we established key relationship between deficiency zero networks and FDE. In particular, we proved that every deficiency zero network satisfies FID $=$ FIID, making it within the class of FDE networks. In addition, it turns out that a network has zero deficiency if and only if its FID is a zero deficiency decomposition (ZDD).

To facilitate classification, we introduce the Deficiency Difference ($\Delta$), measuring the gap between total and subnetwork deficiencies, and the Common Complexes Cardinality ($|\mathscr{CC}|$) of the FID. We summarize the results in Table \ref{summary:DD:CC}.
These metrics support a structured method for determining a network's FDA class. First, by calculating the deficiency difference range to isolate ILC classes and DLC-3; and second, by evaluating the $|\mathscr{CC}|$ value for DLC networks to distinguish between DLC-1 (where $|\mathscr{CC}|=0$) and DLC-2 (where $|\mathscr{CC}|\ge1$).

\begin{table}[h]
\centering
\caption{{Summary of the classification of each FDA class according to deficiency difference and $\delta - |\mathscr{CC}|$-Level}}
\renewcommand{\arraystretch}{1.5} 
\label{summary:DD:CC}
\begin{tabular}{|c|c|c|}
\hline
\textbf{FDA Class} & \textbf{Deficiency Difference} & \textbf{$\delta$ - $|\mathscr{CC}|$-Level} \\ \hline
ILC-1              & $(=, =, \geq)$               & $\delta \geq 0$                      \\ \hline
DLC-1              & $(>, =, >)$                  & $|\mathscr{CC}| = 0$                    \\ \hline
ILC-2              & $(=, >, =)$                  & $\delta \geq 1$                      \\ \hline
DLC-2              & $(>, =, >)$                  & $|\mathscr{CC}| \geq 1$                 \\ \hline
ILC-3              & $(=, >, >)$                  & $\delta \geq 1$                      \\ \hline
DLC-3              & $(>, >, >)$                  & $|\mathscr{CC}| \geq 2$                 \\ \hline
\end{tabular}
\end{table}

Table \ref{computational:roadmap:equilibria} translates the theoretical framework into a practical guide. It confirms that for any ILC network, the LCD is a reliable basis for analyzing both general and complex-balanced equilibria.

The table also shows that for all DLC networks, the LCD is insufficient for general equilibria analysis. For these networks (DLC-1, -2, and -3), one must use the FID to correctly decompose the system for analysis.

\begin{table}[h]
\centering
\caption{{The Computational Road Map for Equilibria Analysis by FDA Class}}
\label{tab:fda_roadmap_revised}
\begin{tabular}{@{}lll@{}}
\toprule
\textbf{FDA Class} & \textbf{To Find General Equilibria} & \textbf{To Find Complex-Balanced} \\ \midrule
\textbf{ILC-1} & Use LCD (or finer ID) & Use LCD (or finer IID) \\
\textbf{ILC-2} & Use LCD (or finer ID) & Use LCD (or finer IID) \\
\textbf{ILC-3} & Use LCD (or finer ID) & Use LCD (or finer IID) \\ \midrule
\textbf{DLC-1} & Must use FID & Use LCD (or finer IID) \\
\textbf{DLC-2} & Must use FID & Must use FIID \\
\textbf{DLC-3} & Must use FID & Use LCD (or finer IID) \\ \bottomrule
\end{tabular}
\label{computational:roadmap:equilibria}
\end{table}

Two promising areas of research based on the FDA framework are identified for future exploration. The first is on the analysis of kinetic subclasses. Future research should focus on subclasses of kinetics, such as power law kinetics, within specific FDA classes. The results could provide refinements for the ILC architectures of theorems valid for all ILC networks or extend the results on ILC networks to handle special kinetics on DLC classes. This is relevant as many biochemical systems belong to DLC classes, as demonstrated by the examples provided in Table \ref{deficiency:difference:cc:cardinalities} and, in general, in Section \ref{section:FDA:DLC}. In particular, the study of the FDE subclass in ILC-1 is particularly interesting, as it encompasses all deficiency zero networks. Next is on network decomposition and class variation. Since the FDA class of a subnetwork can differ from the whole network, this framework is critical for decomposing large networks. For instance, ILC-1 and ILC-2 networks can be used to construct ILC-3 networks, suggesting that identifying FDA classes can facilitate the analysis of intractable architectures by leveraging the properties of their simpler constituent subnetworks.


\section*{Acknowledgement}
BSH acknowledges the Institute of Mathematics, College of Science, University of the
Philippines Diliman for funding support through the Faculty Research Grant.




\appendix

\section{List of abbreviations}
\begin{tabular}{ll}
\noalign{\smallskip}\hline\noalign{\smallskip}
Abbreviation& Meaning \\
\noalign{\smallskip}\hline\noalign{\smallskip}
CRN& chemical reaction network\\
CRNT& chemical reaction network theory\\
DAC& direct air capture\\
DLC& dependent linkage classes\\
FDA& finest decompositions architecture\\
FDC& finest decompositions coarsening\\
FDE& finest decompositions equality\\
FID& finest independent decomposition\\
FIID& finest incidence-independent decomposition\\
ILC& independent linkage classes\\
LCD& linkage class decomposition\\
ODE& ordinary differential equation\\
PLK& power-law kinetics\\
PL-RDK& power-law reactant-determined kinetics\\
PL-NDK& power-law non-reactant-determined kinetics\\
ZDD& zero deficiency decomposition\\
\noalign{\smallskip}\hline
\end{tabular}

\allowdisplaybreaks

\section{Species and reactions of biochemical and environmental systems}
\label{reaction:networks}

\subsection{Insulin metabolic signaling (INSMS)}
\label{appendix:insms}
Species:
\begin{align*}
	& X_2 = \text{Unbound surface insulin receptors (in molar)} \\
	& X_3 = \text{Unphosphorylated once-bound surface receptors (in molar)} \\
	& X_4 = \text{Phosphorylated twice-bound surface receptors (in molar)} \\
	& X_5 = \text{Phosphorylated once-bound surface receptors (in molar)} \\
	& X_6 = \text{Unbound unphosphorylated intracellular receptors (in molar)} \\
	& X_7 = \text{Phosphorylated twice-bound intracellular receptors (in molar)} \\
	& X_8 = \text{Phosphorylated once-bound intracellular receptors (in molar)} \\
	& X_9 = \text{Unphosphorylated IRS-1 (in molar)} \\
	& X_{10} = \text{Tyrosine-phosphorylated IRS-1 (in molar)} \\
	& X_{11} = \text{Unactivated PI 3-kinase (in molar)} \\
	& X_{12} = \text{Tyrosine-phosphorylated IRS-1/activated PI 3-kinase complex (in molar)} \\
	& X_{13} = \text{PI(3,4,5)P$_3$ out of the total lipid population (in \%)} \\
	& X_{14} = \text{PI(4,5)P$_2$ out of the total lipid population (in \%)} \\
	& X_{15} = \text{PI(3,4)P$_2$ out of the total lipid population (in \%)} \\
	& X_{16} = \text{Unactivated Akt (in \%)} \\
	& X_{17} = \text{Activated Akt (in \%)} \\
	& X_{18} = \text{Unactivated PKC-$\zeta$ (in \%)} \\
	& X_{19} = \text{Activated PKC-$\zeta$ (in \%)} \\
	& X_{20} = \text{Intracellular GLUT4 (in \%)} \\
	& X_{21} = \text{Cell surface GLUT4 (in \%)}.
\end{align*}
Reactions:
\begin{multicols}{2}
\noindent
\begin{align*}
	& R_1: X_2 \rightarrow X_3 \\
	& R_2: X_3 \rightarrow X_2 \\
	& R_3: X_5 \rightarrow X_4 \\
	& R_4: X_4 \rightarrow X_5 \\
	& R_5: X_3 \rightarrow X_5 \\
	& R_6: X_5 \rightarrow X_2 \\
	& R_7: X_2 \rightarrow X_6 \\
	& R_8: X_6 \rightarrow X_2 \\
	& R_9: X_4 \rightarrow X_7 \\
	& R_{10}: X_7 \rightarrow X_4 \\
	& R_{11}: X_5 \rightarrow X_8 \\
	& R_{12}: X_8 \rightarrow X_5 \\
	& R_{13}: 0 \rightarrow X_6 \\
	& R_{14}: X_6 \rightarrow 0 \\
	& R_{15}: X_7 \rightarrow X_6 \\
	& R_{16}: X_8 \rightarrow X_6 \\
	& R_{17}: X_9 + X_4 \rightarrow X_{10} + X_4 \\
	& R_{18}: X_9 + X_5 \rightarrow X_{10} + X_5 \\
	& R_{19}: X_{10} \rightarrow X_9 \\
	& R_{20}: X_{10} + X_{11} \rightarrow X_{12} \\
	& R_{21}: X_{12} \rightarrow X_{10} + X_{11} \\
	& R_{22}: X_{14} + X_{12} \rightarrow X_{13} + X_{12} \\
	& R_{23}: X_{13} \rightarrow X_{14} \\
	& R_{24}: X_{15} \rightarrow X_{13} \\
	& R_{25}: X_{13} \rightarrow X_{15} \\
	& R_{26}: X_{16} + X_{13} \rightarrow X_{17} + X_{13} \\
	& R_{27}: X_{17} \rightarrow X_{16} \\
	& R_{28}: X_{18} + X_{13} \rightarrow X_{19} + X_{13} \\
	& R_{29}: X_{19} \rightarrow X_{18} \\
	& R_{30}: X_{20} \rightarrow X_{21} \\
	& R_{31}: X_{21} \rightarrow X_{20} \\
	& R_{32}: X_{20} + X_{17} \rightarrow X_{21} + X_{17} \\
	& R_{33}: X_{20} + X_{19} \rightarrow X_{21} + X_{19} \\
	& R_{34}: 0 \rightarrow X_{20} \\
	& R_{35}: X_{20} \rightarrow 0.
\end{align*}
\end{multicols}

\subsection{Insulin resistance signaling (INRES)}
\label{appendix:inres}
Species:
\begin{align*}
    & X_2 = \text{Inactive insulin receptors} \\
    & X_3 = \text{Insulin-bound receptors} \\
    & X_4 = \text{Tyrosine-phosphorylated receptors} \\
    & X_6 = \text{Internalized dephosphorylated receptors} \\
    & X_7 = \text{Tyrosine-phosphorylated and internalized receptors} \\
    & X_9 = \text{Inactive IRS-1} \\
    & X_{10} = \text{Tyrosine-phosphorylated IRS-1} \\
    & X_{20} = \text{Intracellular GLUT4} \\
    & X_{21} = \text{Cell surface GLUT4} \\
    & X_{22} = \text{Combined tyrosine/serine 307-phosphorylated IRS-1} \\
    & X_{23} = \text{Serine 307-phosphorylated IRS-1} \\
    & X_{24} = \text{Inactive negative feedback} \\
    & X_{25} = \text{Active negative feedback} \\
    & X_{26} = \text{Inactive PKB} \\
    & X_{27} = \text{Threonine 308-phosphorylated PKB} \\
    & X_{28} = \text{Serine 473-phosphorylated PKB} \\
    & X_{29} = \text{Combined threonine 308/serine 473-phosphorylated PKB} \\
    & X_{30} = \text{mTORC1} \\
    & X_{31} = \text{mTORC1 involved in phosphorylation of IRS-1 at serine 307} \\
    & X_{32} = \text{mTORC2} \\
    & X_{33} = \text{mTORC2 involved in phosphorylation of PKB at threonine 473} \\
    & X_{34} = \text{AS160} \\
    & X_{35} = \text{AS160 phosphorylated at threonine 642} \\
    & X_{36} = \text{S6K} \\
    & X_{37} = \text{Activated S6K phosphorylated at threonine 389} \\
    & X_{38} = \text{S6} \\
    & X_{39} = \text{Activated S6 phosphorylated at serine 235 and serine 236} \\
    & X_{40} = \text{ERK} \\
    & X_{41} = \text{ERK phosphorylated at threonine 202 and tyrosine 204} \\
    & X_{42} = \text{ERK sequestered in an inactive pool} \\
    & X_{43} = \text{Elk1} \\
    & X_{44} = \text{Elk1 phosphorylated at serine 383}.
\end{align*}

Reactions:
\begin{multicols}{2}
\noindent
\begin{align*}
    & R_{1}: X_2 \rightarrow X_3 \\
    & R_{2}: X_2 \rightarrow X_4 \\
    & R_{3}: X_3 \rightarrow X_4 \\
    & R_{4}: X_4 \rightarrow X_7 \\
    & R_{5}: X_7 + X_{25} \rightarrow X_6 + X_{25} \\
    & R_{6}: X_4 \rightarrow X_2 \\
    & R_{7}: X_6 \rightarrow X_2 \\
    & R_{8}: X_7 + X_9 \rightarrow X_7 + X_{10} \\
    & R_{9}: X_9 \rightarrow X_{23} \\
    & R_{10}: X_{10} \rightarrow X_9 \\
    & R_{11}: X_{10} + X_{31} \rightarrow X_{22} + X_{31} \\
    & R_{12}: X_{22} \rightarrow X_{10} \\
    & R_{13}: X_{22} \rightarrow X_{23} \\
    & R_{14}: X_{23} \rightarrow X_9 \\
    & R_{15}: X_{10} + X_{24} \rightarrow X_{10} + X_{25} \\
    & R_{16}: X_{25} \rightarrow X_{24} \\
    & R_{17}: X_{10} + X_{26} \rightarrow X_{10} + X_{27} \\
    & R_{18}: X_{27} \rightarrow X_{26} \\
    & R_{19}: X_{27} + X_{33} \rightarrow X_{29} + X_{33} \\
    & R_{20}: X_{22} + X_{28} \rightarrow X_{22} + X_{29} \\
    & R_{21}: X_{29} \rightarrow X_{28} \\
    & R_{22}: X_{28} \rightarrow X_{26} \\
    & R_{23}: X_{29} + X_{30} \rightarrow X_{29} + X_{31} \\
    & R_{24}: X_{27} + X_{30} \rightarrow X_{27} + X_{31} \\
    & R_{25}: X_{31} \rightarrow X_{30} \\
    & R_{26}: X_7 + X_{32} \rightarrow X_7 + X_{33} \\
    & R_{27}: X_{33} \rightarrow X_{32} \\
    & R_{28}: X_{29} + X_{34} \rightarrow X_{29} + X_{35} \\
    & R_{29}: X_{28} + X_{34} \rightarrow X_{28} + X_{35} \\
    & R_{30}: X_{35} \rightarrow X_{34} \\
    & R_{31}: X_{35} + X_{20} \rightarrow X_{35} + X_{21} \\
    & R_{32}: X_{21} \rightarrow X_{20} \\
    & R_{33}: X_{31} + X_{36} \rightarrow X_{31} + X_{37} \\
    & R_{34}: X_{37} \rightarrow X_{36} \\
    & R_{35}: X_{37} + X_{38} \rightarrow X_{37} + X_{39} \\
    & R_{36}: X_{38} + X_{41} \rightarrow X_{39} + X_{41} \\
    & R_{37}: X_{39} \rightarrow X_{38} \\
    & R_{38}: X_7 + X_{40} \rightarrow X_7 + X_{41} \\
    & R_{39}: X_{22} + X_{40} \rightarrow X_{22} + X_{41} \\
    & R_{40}: X_{40} \rightarrow X_{41} \\
    & R_{41}: X_{41} \rightarrow X_{42} \\
    & R_{42}: X_{42} \rightarrow X_{40} \\
    & R_{43}: X_{41} + X_{43} \rightarrow X_{41} + X_{44} \\
    & R_{44}: X_{44} \rightarrow X_{43}
\end{align*}
\end{multicols}

\subsection{Wnt signaling}
\label{appendix:Wnt}
Species:
\begin{align*}
    &A_1= \text{destruction complex (DC) (active form)} \\
    &A_2=\text{DC (inactive form)} \\
    &A_3=\text{active DC residing in the nucleus} \\
    & A_4= \beta\text{-catenin} \\
    & A_5= \beta\text{-catenin in the nucleus} \\
    & A_6= \text{T-cell factor (TCF)} \\
    & A_7=  \beta\text{-catenin-TCF complex} \\
    & A_8= \beta\text{-catenin bound with DC} \\
    & A_9= \beta\text{-catenin bound with DC in the nucleus} \\
    & A_{10}=  \beta\text{-catenin (for proteasomal degradation)} \\
    & A_{11}=  \beta\text{-catenin (for proteasomal degradation) in the nucleus} \\
    & A_{12}=   \text{dishevelled (inactive form)} \\
    & A_{13}=  \text{dishevelled (active form)} \\
    &A_{14}=  \text{active dishevelled in the nucleus} \\
    & A_{15}=   \text{inactive DC in the nucleus} \\
    & A_{16}=  \text{phosphatase} \\
    & A_{17}=  \text{phosphatase in the nucleus} \\
    & A_{18}=  \text{active DC bound with dishevelled} \\
    & A_{19}=  \text{active DC bound with dishevelled in the nucleus} \\
    & A_{20}= \text{active DC bound with phosphatase} \\
    & A_{21}=  \text{active DC bound with phosphatase in the nucleus} \\
    & A_{22}=  \text{GSK3}\beta \\
    & A_{23}=  \text{axin-APC complex} \\
    & A_{24}=  \text{APC} \\
    & A_{25}=  \beta\text{-catenin bound with DC (for proteasomal degradation)} \\
    & A_{26}=  \text{axin} \\
    & A_{27}=  \beta\text{-catenin-axin complex} \\
    & A_{28}=  \text{a complex considered as a single species}:
    (A_{13}+A_{22}+A_{23}=A_{28})
\end{align*}
Reactions of Feinberg-Lee Wnt Signaling:
\begin{multicols}{2}
\noindent
\begin{align*}
& R_1: 0 \rightarrow A_4 \\
& R_{2}: A_4 \rightarrow 0 \\
& R_3: A_1 + A_4 \rightarrow A_8 \\
& R_4: A_8 \rightarrow A_1  +  A_4 \\
& R_{5}: A_{10} \rightarrow 0 \\
& R_{6}: A_1 \rightarrow A_2 \\
& R_{7}: A_2 \rightarrow A_1 \\
& R_{8}: A_{12} \rightarrow A_{13} \\
& R_{9}: A_{13} \rightarrow A_{12} \\
& R_{10}: A_{24} + A_{26} \rightarrow A_{23} \\
& R_{11}: A_{23} \rightarrow A_{24} + A_{26} \\
& R_{12}: A_8 \rightarrow A_{25} \\
& R_{13}: A_{25} \rightarrow A_1 + A_{10} \\
& R_{14}: 0 \rightarrow A_{26} \\
& R_{15}: A_{26} \rightarrow 0 \\
& R_{16}: A_4 + A_6 \rightarrow A_7 \\
& R_{17}: A_7 \rightarrow A_4 + A_6 \\
& R_{18}: A_{24} + A_4 \rightarrow A_{27} \\
& R_{19}: A_{27} \rightarrow A_{24} + A_4 \\
& R_{20}: A_{13} + A_2 \rightarrow A_{28} \\
& R_{21}: A_2 \rightarrow A_{23} \\
& R_{22}: A_{23} \rightarrow A_2 \\
& R_{23}: A_{28} \rightarrow A_{13} + A_{23}
\end{align*}
\end{multicols}
Reactions of Schmitz Wnt Signaling:
\begin{multicols}{2}
\noindent
\begin{align*}
& R_1: 0 \rightarrow A_4 \\
& R_2: A_4 \rightarrow A_5 \\
& R_3: A_5 \rightarrow A_4 \\
& R_4: A_1 + A_4 \rightarrow A_8 \\
& R_5: A_8 \rightarrow A_1  +  A_4 \\
& R_6: A_5 + A_3 \rightarrow A_9 \\
& R_7: A_9 \rightarrow A_5  +  A_3 \\
& R_8: A_6 + A_5 \rightarrow A_7 \\
& R_9: A_7 \rightarrow A_6  +  A_5 \\
& R_{10}: A_8 \rightarrow A_1 + A_{10} \\
& R_{11}: A_9 \rightarrow A_3 + A_{11} \\
& R_{12}: A_{10} \rightarrow 0 \\
& R_{13}: A_{11} \rightarrow 0 \\
& R_{14}: A_1 \rightarrow A_2 \\
& R_{15}: A_2 \rightarrow A_1 \\
& R_{16}: A_1 \rightarrow A_3 \\
& R_{17}: A_3 \rightarrow A_1
\end{align*}
\end{multicols}
Reactions of MacLean Wnt Signaling:
\begin{multicols}{2}
\noindent
\begin{align*}
& R_1: 0 \rightarrow A_4 \\
& R_2: A_4 \rightarrow A_5 \\
& R_3: A_5 \rightarrow A_4 \\
& R_4: A_1 + A_4 \rightarrow A_8 \\
& R_5: A_8 \rightarrow A_1  +  A_4 \\
& R_6: A_5 + A_3 \rightarrow A_9 \\
& R_7: A_9 \rightarrow A_5  +  A_3 \\
& R_8: A_6 + A_5 \rightarrow A_7 \\
& R_9: A_7 \rightarrow A_6  +  A_5 \\
& R_{10}: A_{12} \rightarrow A_{13} \\
& R_{11}: A_{13} \rightarrow A_{12} \\
& R_{12}: A_{13} \rightarrow A_{14} \\
& R_{13}: A_{14} \rightarrow A_{13} \\
& R_{14}: A_2 \rightarrow A_{15} \\
& R_{15}: A_{15} \rightarrow A_2 \\
& R_{16}: A_3  +  A_{14} \rightarrow A_{19} \\
& R_{17}: A_{19} \rightarrow A_3  +  A_{14} \\
& R_{18}: A_{19} \rightarrow A_{14}  +  A_{15} \\
& R_{19}: A_{15}  +  A_{17} \rightarrow A_{21} \\
& R_{20}: A_{21} \rightarrow A_{15}  +  A_{17} \\
& R_{21}: A_{21} \rightarrow A_3  +  A_{17} \\
& R_{22}: A_{13}  +  A_1 \rightarrow A_{18} \\
& R_{23}: A_{18} \rightarrow A_{13}  + A_1 \\
& R_{24}: A_{18} \rightarrow A_{13}  + A_2 \\
& R_{25}: A_2  +  A_{16} \rightarrow A_{20} \\
& R_{26}: A_{20} \rightarrow A_2  +  A_{16} \\
& R_{27}: A_{20} \rightarrow A_1  +  A_{16} \\
& R_{28}: A_8 \rightarrow A_1 \\
& R_{29}: A_9 \rightarrow A_3 \\
& R_{30}: A_4 \rightarrow 0 \\
& R_{31}: A_5 \rightarrow 0
\end{align*}
\end{multicols}

\subsection{Reactions of the Mycobacterium tuberculosis NRP S-system}
\label{appendix:tb:nrp}
\begin{align*}
&R_{1}:  X_2  \to X_2 + X_1\\
&R_{2}:  X_1 \to 0\\
&R_{3}:  X_1  \to X_1 + X_2\\
&R_{4}:  X_2 \to 0\\
&R_{5}:  X_1 + X_{13} + X_7 + X_6 + X_4 + X_{16} + X_{11} + X_{38}  \to X_1 + X_{13} + X_7 + X_6 + X_4\\ 
& \quad + X_{16} + X_{11} + X_{38} + X_3\\
&R_{6}:  X_3 \to 0\\
&R_{7}:  X_3 + X_{11} + X_6 + X_{12} + X_9 + X_5 + X_{10}  \to X_3 + X_{11} + X_6 + X_{12} + X_9 + X_5 \\ 
& \quad + X_{10} + X_4\\
&R_{8}:  X_4 \to 0\\
&R_{9}:  X_4 + X_6 + X_9 + X_{11} + X_{12}  \to X_4 + X_6 + X_9 + X_{11} + X_{12} + X_5\\
&R_{10}: X_5 \to 0\\
&R_{11}: X_4 + X_5 + X_9 + X_3 + X_{11} + X_{12} + X_{10}  \to X_4 + X_5 + X_9 + X_3 + X_{11} + X_{12}\\ 
& \quad  + X_{10} + X_6\\
&R_{12}: X_6 \to 0\\
&R_{13}: X_3 + X_8 + X_{40}  \to X_3 + X_8 + X_{40} + X_7\\
&R_{14}: X_7 \to 0\\
&R_{15}: X_1 + X_7 + X_{21} + X_{22} + X_{24} + X_{15}  \to X_1 + X_7 + X_{21} + X_{22} + X_{24} + X_{15} + X_8\\
&R_{16}: X_8 \to 0\\
&R_{17}: X_5 + X_4 + X_6 + X_{10} + X_{11} + X_{12}  \to X_5 + X_4 + X_6 + X_{10} + X_{11} + X_{12} + X_9\\
&R_{18}: X_9 \to 0\\
&R_{19}: X_4 + X_6 + X_9  \to X_4 + X_6 + X_9 + X_{10}\\
&R_{20}: X_{10}  \to 0\\
&R_{21}: X_3 + X_4 + X_5 + X_6 + X_9 + X_{12} + X_{17} + X_{13}  \to X_3 + X_4 + X_5 + X_6 + X_9 \\ 
& \quad + X_{12} + X_{17} + X_{13} + X_{11}\\
&R_{22}: X_{11}  \to 0\\
&R_{23}: X_4 + X_5 + X_6 + X_9 + X_{11} + X_{17} + X_{13}  \to  X_4 + X_5 + X_6 + X_9 + X_{11} + X_{17}\\ 
& \quad  + X_{13} + X_{12}\\
&R_{24}: X_{12}  \to 0\\
&R_{25}: X_3 + X_1 + X_{11} + X_{18} + X_{16} + X_{36}  \to X_3 + X_1 + X_{11} + X_{18} + X_{16} + X_{36} + X_{13}\\
&R_{26}: X_{13}  \to 0\\
&R_{27}: X_1 + X_2 + X_{15}  \to X_1 + X_2 + X_{15} + X_{14}\\
&R_{28}: X_{14}  \to 0\\
&R_{29}: X_1 + X_8 + X_{14}  \to X_1 + X_8 + X_{14} + X_{15}\\
&R_{30}: X_{15}  \to 0\\
&R_{31}: X_1 + X_3 + X_{13}  \to X_1 + X_3 + X_{13} + X_{16}\\
&R_{32}: X_{16}  \to 0\\
&R_{33}: X_{11} + X_{12}  \to X_{11} + X_{12} + X_{17}\\
&R_{34}: X_{17}  \to 0\\
&R_{35}: X_{13} + X_{36}  \to X_{13} + X_{36} + X_{18}\\
&R_{36}: X_{18}  \to 0\\
&R_{37}: X_7  \to X_7 + X_{19}\\
&R_{38}: X_{19}  \to 0\\
&R_{39}: X_7  \to X_7 + X_{20}\\
&R_{40}: X_{20}  \to 0\\
&R_{41}: X_8  \to X_8 + X_{21}\\
&R_{42}: X_{21}  \to 0\\
&R_{43}: X_1 + X_8  \to X_1 + X_8 + X_{22}\\
&R_{44}: X_{22}  \to 0 \\
&R_{45}: X_1  \to X_1 + X_{23}\\
&R_{46}: X_{23}  \to 0\\
&R_{47}: X_1 + X_8  \to X_1 + X_8 + X_{24}\\
&R_{48}: X_{24}  \to 0\\
&R_{49}: X_1  \to X_1 + X_{25}\\
&R_{50}: X_{25}  \to 0\\
&R_{51}: X_1  \to X_1 + X_{26}\\
&R_{52}: X_{26}  \to 0\\
&R_{53}: X_1 + X_2  \to X_1 + X_2 + X_{27}\\
&R_{54}: X_{27}  \to 0\\
&R_{55}: X_1 + X_2  \to X_1 + X_2 + X_{28}\\
&R_{56}: X_{28}  \to 0\\
&R_{57}: X_2  \to X_2 + X_{29}\\
&R_{58}: X_{29}  \to 0\\
&R_{59}: X_2  \to X_2 + X_{30}\\
&R_{60}: X_{30}  \to 0\\
&R_{61}: X_2  \to X_2 + X_{31}\\
&R_{62}: X_{31}  \to 0\\
&R_{63}: X_2  \to X_2 + X_{32}\\
&R_{64}: X_{32}  \to 0\\
&R_{65}: X_2  \to X_2 + X_{33}\\
&R_{66}: X_{33}  \to 0\\
&R_{67}: X_2  \to X_2 + X_{34}\\
&R_{68}: X_{34}  \to 0\\
&R_{69}: X_2  \to X_2 + X_{35}\\
&R_{70}: X_{35}  \to 0\\
&R_{71}: X_1 + X_{13} + X_{18}  \to X_1 + X_{13} + X_{18} + X_{36}\\
&R_{72}: X_{36}  \to 0\\
&R_{73}: X_2  \to X_2 + X_{37}\\
&R_{74}: X_{37}  \to 0\\
&R_{75}: X_3  \to X_3 + X_{38}\\
&R_{76}: X_{38}  \to 0\\
&R_{77}: X_4  \to X_4 + X_{39}\\
&R_{78}: X_{39}  \to 0\\
&R_{79}: X_1 + X_7  \to X_1 + X_7 + X_{40}\\
&R_{80}: X_{40}  \to 0
\end{align*}

\subsection{Reactions of the model of a complex coordination of multi-scale cellular responses to environmental stress}
\label{appendix:complex:coordination}
\begin{multicols}{2}
\noindent
\begin{align*}
&R_1: X_1+X_3 \to X_2+X_3\\
&R_2: X_2+X_7 \to X_3+X_7\\
&R_3: X_3 \to X_4\\
&R_4: X_4 \to X_3\\
&R_5: X_4 \to X_5\\
&R_6: X_5 \to X_4\\
&R_7: X_5+X_3 \to X_6+X_3\\
&R_8: X_6+X_5+X_3 \to X_4 + X_5 + X_3\\
&R_9: X_5+X_3+X_2 \to X_7 + X_2\\
&R_{10}: X_7\to X_8\\
&R_{11}: X_8\to X_2\\
&R_{12}: X_3\to X_9\\
&R_{13}: X_9\to X_{10}\\
&R_{14}: X_9\to X_{11}\\
&R_{15}: X_3\to X_{12}
\end{align*}
\end{multicols}

\section{The computational package DECENT}
\label{DECENT:Anderies}

\subsection{Instruction}

\begin{spverbatim}
To fill out the `model' structure, write a string for `model.id'. This is just to put a name to the network. To add the reactions to the network, use the function addReaction where the output is `model'. 

The addReaction function
      - OUTPUT: Returns a structure called `model' with added field `reaction' with subfields `id', `reactant', `product', `reversible', and `kinetic'. The output variable `model' allows the user to view the network with the added reaction.
      - INPUTS:
           - model: a structure, representing the CRN
           - id: visual representation of the reaction, e.g.,
             reactant -> product (string)
           - reactant_species: species of the reactant complex (cell)
           - reactant_stoichiometry: stoichiometry of the species of the
             reactant complex (cell)
           - reactant_kinetic: kinetic orders of the species of the
             reactant complex (array)
           - product_species: species of the product complex (cell)
           - product_stoichiometry: stoichiometry of the species of the
             product complex (cell)
           - product_kinetic: "kinetic orders" of the species of the product
             complex, if the reaction is reversible (array); if the reaction
             in NOT reversible, leave blank
           - reversible: logical; whether the reaction is reversible or not
             (true or false)
      * Make sure the function addReaction is in the same folder/path being
        used as the current working directory.
\end{spverbatim}

\subsection{A script for the Anderies et al. network}
We provide a sample script as follows. To compute the FIID, just replace the FID by FIID in the last line of the script.
\begin{verbatim}
% Input the chemical reaction network 
model.id = `Example 1';
model = addReaction(model, `A1+2A2->2A1+A2', ...           
                           {`A1',`A2'}, {1,2}, [1,2], ... 
                           {`A1',`A2'}, {2,1}, [ ], ...   
                           false);                        
model = addReaction(model, `A1+A2->2A2', ...            
                           {`A1',`A2'}, {1,1}, [1,1], ...
                           {`A2'}, {2}, [ ], ...
                           false);  
model = addReaction(model, `A2<->A3', ...
                           {`A2'}, {1}, [1], ...
                           {`A3'}, {1}, [1], ...
                           true);
% Generate the finest incidence independent decomposition
[model, I_a, G, P, results] = FID_check(model);
\end{verbatim}

\subsection{Output}
\begin{verbatim}
The finest independent decomposition is:

N1: R1, R2 
N2: R3, R4 

Deficiencies:

      Characteristic      Notation    Network    |    N1    N2    Sum
    ___________________    ________    _______    _    __    __    ___

    Deficiency             delta          1       |    1     0      1 
    
The finest incidence independent decomposition is:

N1: R1 
N2: R2 
N3: R3, R4 

Deficiencies:

      Characteristic      Notation    Network    |    N1    N2    N3    Sum
    ___________________    ________    _______    _    __    __    __    ___

    Deficiency             delta          1       |    0     0     0      0 
\end{verbatim}

\section{Other matrix theory construction for deficiency zero networks} \label{app:mattheo}

For the succeeding results, let $M_{m,n}$ be the set of all $m\times n$ matrices with complex entries and $I_n$ be the $n\times n$ identity matrix. Let $A\in M_{m,n}$. We denote by $A^\top\in M_{n,m}$ the transpose of $A$.
\begin{lemma}\label{prop:eqrank}
    Let $N\in M_{m,r}$, $Y\in M_{m,n}$ and $W\in M_{n,r}$ be the stoichiometric, molecularity and incidence matrices of a CRN, respectively. Then, \[\textup{rank }N\leq\textup{rank }W.\] Moreover, if $Y$ has a left generalized inverse (i.e., $Y'Y=I_n$ for some $Y'\in M_{n,m}$), then \[\textup{rank }N=\textup{rank }W.\]
\end{lemma}
\begin{proof}
    Recall that $N=YW$. Hence, ${\rm{rank \ }} N \le \min\{{\rm{rank \ }} Y,{\rm{rank \ }} W\} \le {\rm{rank \ }} W.$ The second assertion follows from the following:
    \[\textup{rank }I=\textup{rank }Y'YI\leq\textup{rank }YI\leq\textup{rank }I.\] \qed
\end{proof}

\begin{proposition} \cite{tan} \label{prop:tan} Let $A\in M_{m,n}$. Then, $A$ has a left generalized inverse $A'$ if and only if $\textup{rank }A=n$ and $n\leq m$.
\end{proposition}

\begin{theorem} \label{thm:suffcon}
Let $n$ be the number of complexes and $m$ be the number of species in a CRN. Let $N\in M_{m,r},Y\in M_{m,n},W\in M_{n,r}$ be the stoichiometric, molecularity, and incidence matrices of the CRN, respectively. If $\textup{rank }Y=n$ and $n\leq m$, then the CRN is an FDE network. Consequently, if $Y^\top Y$ is invertible, then the CRN is an FDE network. 
\end{theorem}
\begin{proof}
    It follows from Proposition \ref{prop:tan} that $Y$ has a left generalized inverse $Y'\in M_{n,m}$ and so, $Y'Y=I_n$. Suppose that $\textup{rank }W=p$. Note that $\textup{rank }N=\textup{rank }YW=\textup{rank }W$ by Lemma \ref{prop:eqrank}. Recall that each column $n_j$ and $k_j$ of the stoichiometric and incidence matrices, respectively, corresponds to the same reaction vector $R_j$ of the CRN. Let $\{n_{i_1},...,n_{i_p}\}$ be a linearly independent set of columns of $N$, which is maximal since $\textup{rank }N=p$. Since \[N=[n_1 \hdots n_r]=[Yw_1\hdots Yw_r]=YW,\] we have that $\{Yw_{i_1},...,Yw_{i_p}\}$ is also a maximal linearly independent set. Thus, $a_{i_1}Yw_{i_1}+\hdots+a_{i_r}Yw_{i_p}=0$ if and only if $a_{i_1}=\hdots=a_{i_p}=0$. Since $Y'Y=I_n$, we have that $a_{i_1}w_{i_1}+\hdots+a_{i_r}w_{i_p}=0$ if and only if $a_{i_1}=\hdots=a_{i_p}=0$. Therefore, $\{w_{i_1},...,w_{i_p}\}$ is a linearly independent set of columns of $W$, which is also maximal since $\textup{rank }W=p$. For each vector $n_k$ distinct from $\{n_{i_1},...,n_{i_p}\}$, we may write $n_k=\sum\limits_{j=1}^p\alpha_{k_j}n_{i_j}$. Since $n_k=Yw_k$ for all $k\in\{1,...,r\}$,
\begin{equation}\label{lincomb}
    Yw_k=n_k=\sum\limits_{j=1}^p\alpha_{k_j}n_{i_j}=\sum\limits_{j=1}^p\alpha_{k_j}Yw_{i_j}.
\end{equation}

 Multiplying $Y'$ on the left side of \eqref{lincomb} gives 
\[
 w_k=\sum\limits_{j=1}^p\alpha_{k_j}w_{i_j}.
\]
Therefore, $n_k=\sum\limits_{j=1}^p\alpha_{k_j}n_{i_j}$ if and only if $w_k=\sum\limits_{j=1}^p\alpha_{k_j}w_{i_j}$ for all $n_k$ distinct from  $\{n_{i_1},...,n_{i_p}\}$ and all $w_k$ distinct from $\{w_{i_1},...,w_{i_p}\}$. Hence, by the five step method in finding independent (incidence independent) decompositions of CRNs in \cite{hernandez:delacruz1}, the columns of $N$ and $W$ forms the same coordinate graph in which it follows that the same corresponding CRN decompositions hold. Thus, $\textup{FID}=\textup{FIID}$. 

Suppose that $Y^\top Y$ is invertible and let $Y'=(Y^\top Y)Y^\top$. Then, \[Y'Y=(Y^\top Y)^{-1}Y^\top Y= I_n.\] Hence, the invertibility of $Y^\top Y$ is sufficient for the existence of a left generalized inverse $Y'$ of $Y$. By Proposition \ref{prop:tan}, $\text{rank }Y=n$ and $n\leq m$. By the preceding, $\text{FID}=\text{FIID}$.
\qed
\end{proof}

\begin{example}
    Consider the following CRN: $R_1: A\to B$, $R_2:B\to A$, and $R_3:C+D\to D+E$. The corresponding molecularity matrix is given by \[
    Y=\begin{blockarray}{ccccc}
        \matindex{$A$} & \matindex{$B$} & \matindex{$C+D$} & \matindex{$D+E$} \\
        \begin{block}{[cccc]c}
        1 & 0 & 0 & 0  & \matindex{$A$}\\
        0 & 1 & 0 & 0 & \matindex{$B$}\\
        0 & 0 & 1 & 0 & \matindex{$C$}\\
        0 & 0 & 1 & 1 & \matindex{$D$}\\
        0 & 0 & 0 & 1 & \matindex{$E$}\\
        \end{block}
        \end{blockarray}.
    \]
    Note that $\textup{rank }Y=4$, which is the number of complexes. Furthermore, the number of complexes is less than the number of species. By Proposition \ref{prop:tan}, a left generalized inverse of $Y$ exists and is given by
    \[
    Y'=\begin{bmatrix}
        1 & 0 & 0 & 0 & 0\\
        0 & 1 & 0 & 0 & 0\\
        0 & 0 & \frac{2}{3} & \frac{1}{3} & -\frac{1}{3}\\
        0 & 0 & -\frac{1}{3} & \frac{1}{3} & \frac{2}{3}
    \end{bmatrix}.
    \]
    The corresponding incidence and stoichiometric matrices, respectively, are as follows:
    \begin{center}
        $W=\begin{blockarray}{cccc}
            \matindex{$R_1$} & \matindex{$R_2$} & \matindex{$R_3$}\\
            \begin{block}{[ccc]c}
                -1 & 1 & 0 & \matindex{$A$}\\
                1 & -1 & 0 & \matindex{$B$}\\
                0 & 0 & -1 & \matindex{$C+D$}\\
                0 & 0 & 1 & \matindex{$D+E$}\\
            \end{block}
        \end{blockarray}$ and $N=\begin{blockarray}{cccc}
            \matindex{$R_1$} & \matindex{$R_2$} & \matindex{$R_3$}\\
            \begin{block}{[ccc]c}
            -1 & 1 & 0 & \matindex{$A$}\\
            1 & -1 & 0 & \matindex{$B$}\\
            0 & 0 & -1 & \matindex{$C$}\\
            0 & 0 & 0 & \matindex{$D$}\\
            0 & 0 & 1 & \matindex{$E$}\\
            \end{block}
        \end{blockarray}$.
    \end{center}
    Observe that $\textup{rank }W=\textup{rank }N=2$ with $\{n_1,n_3\}$ and $\{w_1,w_3\}$ as a maximal linearly independent set of columns for $N$ and $W$, respectively. Moreover, $w_2=-w_1$ and $n_2=-n_1$. Hence, following the five-step method in finding independent (incidence independent) decompositions of CRNs in \cite{hernandez:delacruz1}, we see that the FID consists of $\{R_1,R_2\}$ and $\{R_3\}$ while the FIID consists of $\{R_1,R_2\}$ and $\{R_3\}$. Therefore, FID=FIID.
    This CRN also verifies the assertion in Theorem \ref{thm:suffcon}. A computation yields \[Y^\top Y=\begin{bmatrix}
        1 & 0 & 0 & 0\\
        0 & 1 & 0 & 0\\
        0 & 0 & 2 & 1\\
        0 & 0 & 1 & 2\\
    \end{bmatrix},\] which is invertible. Hence, Theorem \ref{thm:suffcon} implies $\textup{FID}=\textup{FIID}$.
\end{example}

\begin{example}
    Consider the translated Anderies CRN given by Example \ref{example:Anderies2.1}. Observe that the molecularity matrix is $Y=I_3$ and hence, $\textup{rank }Y=3$, which is the number of complexes. Moreover, the number of complexes is equal to the number of species. By Theorem \ref{thm:suffcon}, $\textup{FID}=\textup{FIID}$ which verifies the conclusion from Example \ref{example:Anderies2.1}.
\end{example}

\begin{example}
    A similar example from the translated Anderies CRN is given by the Schmitz's pre-industrial carbon cycle model. The molecularity matrix is $Y=I_6$ and so, we have that $\textup{FID}=\textup{FIID}$ by Theorem \ref{thm:suffcon}.
\end{example}

Theorem \ref{thm:suffcon} result generalizes the case of monomolecular networks.

\bibliographystyle{spmpsci}

\end{document}